\documentclass[twocolumns,10pt,final]{IEEEtran}

\usepackage{times}
\usepackage{amsmath,dsfont}
\usepackage{amssymb,amsthm}
\usepackage{epsfig,verbatim}
\usepackage{mathrsfs}
\usepackage{ifthen}
\usepackage{syntonly}
\usepackage{graphicx}
\usepackage{epstopdf}
\usepackage{subfig}
\usepackage[font=footnotesize]{caption}
\usepackage{multirow}
\usepackage{stfloats}

\newboolean{IncludeJntOldProofs}
\setboolean{IncludeJntOldProofs}{true}

\newboolean{SquizFlag}
\setboolean{SquizFlag}{true}

\newboolean{X3dependOnQ}
\setboolean{X3dependOnQ}{true}

\newcounter{MYtempeqncnt}

\theoremstyle{remark} \newtheorem{theorem}{Theorem}
\theoremstyle{remark} 
\newtheorem{proposition}{Proposition}

\theoremstyle{remark} \newtheorem{remark}{Remark}
\newtheorem*{corollaryA}{Corollary}

\newcommand{\eps}{\epsilon}
\newcommand{\styp}{A_{\epsilon}^{*(n)}}
\newcommand{\stypt}{A_{\epsilon_2}^{*(n)}}
\newcommand{\stypm}{A^{*(m)}_{\epsilon}}
\newcommand{\mB}{\mathcal{B}}
\newcommand{\mX}{\mathcal{X}}
\newcommand{\mY}{\mathcal{Y}}
\newcommand{\mW}{\mathcal{W}}
\newcommand{\mV}{\mathcal{V}}
\newcommand{\mS}{\mathcal{S}}
\newcommand{\mT}{\mathcal{T}}
\newcommand{\mQ}{\mathcal{Q}}
\newcommand{\mU}{\mathcal{U}}
\newcommand{\mCN}{\mathcal{CN}}
\newcommand{\tH}{\tilde{H}}
\newcommand{\tU}{\tilde{U}}
\newcommand{\tXvec}{\tilde{\mathbf{X}}}
\newcommand{\txvec}{\tilde{\mathbf{x}}}
\newcommand{\hXvec}{\hat{\mathbf{X}}}
\newcommand{\hxvec}{\hat{\mathbf{x}}}
\newcommand{\yvec}{\mathbf{y}}
\newcommand{\xvec}{\mathbf{x}}
\newcommand{\zvec}{\mathbf{z}}
\newcommand{\svec}{\mathbf{s}}
\newcommand{\avec}{\mathbf{a}}
\newcommand{\tvec}{\mathbf{t}}
\newcommand{\Zvec}{\mathbf{Z}}

\newcommand{\Qvec}{\mathbf{Q}}
\newcommand{\vvec}{\mathbf{v}}
\newcommand{\wvec}{\mathbf{w}}
\newcommand{\qvec}{\mathbf{q}}
\newcommand{\tsvec}{\tilde{\svec}}
\newcommand{\hsvec}{\hat{\svec}}
\newcommand{\ttvec}{\tilde{\tvec}}
\newcommand{\htvec}{\hat{\tvec}}
\newcommand{\tqvec}{\tilde{\qvec}}
\newcommand{\bqvec}{\bar{\qvec}}
\newcommand{\hqvec}{\hat{\qvec}}
\newcommand{\msg}{u}
\newcommand{\msgBig}{\MakeUppercase{\msg}}
\newcommand{\msgCal}{\mathcal{\msgBig}}
\newcommand{\tmsg}{\tilde{\msg}}
\newcommand{\hmsg}{\hat{\msg}}
\newcommand{\Xvec}{\mathbf{X}}
\newcommand{\Yvec}{\mathbf{Y}}
\newcommand{\Vvec}{\mathbf{V}}
\newcommand{\E}{\mathds{E}}
\newcommand{\CN}{\mathcal{CN}}
\newcommand{\msA}{\mathscr{A}}
\newcommand{\msD}{\mathscr{D}}
\newcommand{\setR}{\mathfrak{R}}
\newcommand{\setC}{\mathfrak{C}}

\newcommand{\chR}{\hat{R}}
\newcommand{\LL}[1]{\left#1}
\newcommand{\RR}[1]{\right#1}
\newcommand{\dst}{\displaystyle}
\newcommand\independent{\protect\mathpalette{\protect\independenT}{\perp}}
\def\independenT#1#2{\mathrel{\rlap{$#1#2$}\mkern2mu{#1#2}}}
\long\def\symbolfootnote[#1]#2{\begingroup\def\thefootnote{\fnsymbol{footnote}}\footnote[#1]{#2}\endgroup}

\newcommand{\thmlabel}[1]{\label{thm:#1}}
\newcommand{\thmref}[1]{\ref{thm:#1}}
\newcommand{\Thmref}[1]{Thm.~\thmref{#1}}

\newtheorem{MyDefinition}{Definition}

 \IEEEoverridecommandlockouts

\title{Source-Channel Coding Theorems for the Multiple-Access Relay Channel
\thanks{
Yonathan Murin and Ron Dabora are with the Department of
Electrical and Computer Engineering, Ben-Gurion University, Israel; Email: {\tt  \{moriny,ron\}@ee.bgu.ac.il}.
 Deniz G\"und\"uz is with the Department of Electrical and Electronic Engineering, Imperial College London, London, United Kingdom; Email: {\tt d.gunduz@imperial.ac.uk}.
This work was partially supported by the European Commission's Marie Curie IRG Fellowship
PIRG05-GA-2009-246657 under the Seventh Framework Programme. Deniz G\"{u}nd\"{u}z was partially supported by the Spanish Government under project TEC2010-17816 (JUNTOS) and the European Commission's Marie Curie IRG Fellowship with reference number 256410 (COOPMEDIA). Parts of this work were presented at the International Symposium on Wireless Communication Systems (ISWCS), Aachen, Germany, November
2011, and at the International Symposium on Information Theory (ISIT), Boston, MA, July 2012.
}
}

\author{Yonathan Murin, Ron Dabora, and Deniz G\"und\"uz} 

\begin{document}
\maketitle
\thispagestyle{empty}

\begin{abstract}
    We study reliable transmission of arbitrarily correlated sources over multiple-access relay channels (MARCs) and multiple-access broadcast relay channels (MABRCs). In MARCs only the destination is interested in reconstructing the sources, while in MABRCs both the relay and the destination want to reconstruct them. In addition to arbitrary correlation among the source
    signals at the users, both the relay and the destination have side information correlated with the source signals. Our objective is to determine whether a given pair of sources can be
    losslessly transmitted to the destination for a given number of channel symbols per source sample, defined as the source-channel rate.
    Sufficient conditions for reliable communication based on operational separation, as well as necessary conditions on the achievable source-channel
    rates are characterized. Since operational separation is generally not optimal for MARCs and MABRCs, sufficient conditions for reliable communication
    using joint source-channel coding schemes based on a combination of the correlation preserving mapping technique with Slepian-Wolf source coding are
    also derived. For correlated sources transmitted over fading Gaussian MARCs and MABRCs, we present conditions under which separation (i.e., separate and stand-alone source and channel codes) is optimal.
    This is the first time optimality of separation is proved for MARCs and MABRCs.
\end{abstract}


\section{Introduction}

The multiple-access relay channel (MARC) models a network in which several users communicate with a single destination with the help of a relay \cite{Kramer:00}.
The MARC is a fundamental multi-terminal channel model that generalizes both the multiple access channel (MAC) and the relay channel models, and has received a lot of attention in the recent years \cite{Kramer:00},
\cite{Kramer:2005}, \cite{Sankar:07}, \cite{Tandon:CISS:11}. If the relay terminal also wants to decode the source messages, the model is called the multiple-access broadcast relay channel (MABRC).

Previous work on MARCs considered independent sources at the terminals. In the present work we allow arbitrary correlation among the sources to be transmitted to the destination in a lossless fashion, and also let the
relay and the destination have side information correlated with the sources. Our objective is to determine whether a given pair of sources can be losslessly transmitted to the destination for a specific number of channel
uses per source sample, which is defined as the {\it source-channel rate}.

In \cite{Shannon:48} Shannon showed that a source can be reliably transmitted over a  point-to-point memoryless channel, if its entropy is less than the capacity of the channel. Conversely, if the source entropy is greater
than the channel capacity, reliable transmission of the source over the channel is not possible.
 Hence, a simple comparison of the rates of the optimal source code and the optimal channel code for the respective source and channel, suffices to determine whether reliable communication is feasible or not. This is called
 the \emph{separation theorem}. An implication of the separation theorem is that the independent design of the source and the channel codes is optimal. However, the optimality of source-channel separation does not generalize
 to multiuser networks \cite{Shannon:61}, \cite{Cover:80}, \cite{GunduzErkip:09}, and, in general the source and the channel codes need to be designed jointly for every particular combination of sources and~channel.

The fact that the MARC generalizes both the MAC and the relay channel models reveals the difficulty of the problem studied here. The capacity of the relay channel, which corresponds to a special case of our problem, is
still unknown. While the capacity region of a MAC is known in general, the optimal joint source-channel code for  transmission of correlated sources over the MAC remains open \cite{Cover:80}.
Accordingly, the objective of this work is to construct lower and upper bounds for the achievable source-channel rates in MARCs and MABRCs. We shall focus on decode-and-forward (DF) based achievability schemes,
such that the relay terminal decodes both source signals before sending cooperative information to the destination. Naturally, DF-based achievable schemes for the MARC directly apply to the MABRC model as well.
Moreover, we characterize the {\em optimal} source-channel rate in some special cases. Our contributions are listed below:

    1) We establish an achievable source-channel rate for MARCs based on operational separation \cite[Section I]{Tuncel:06}.
    The scheme uses the DF strategy with irregular encoding \cite{CoverG:79}, \cite[Section I-A]{Kramer:2005}, successive decoding at the relay
    and backward decoding at the destination. We show that for MARCs with correlated sources and side information, DF with irregular encoding yields a higher
    achievable source-channel rate than the rate achieved by DF with regular encoding. This is in contrast to the scenario without side information, in which
    DF with regular encoding achieve the same source-channel rate as DF with irregular encoding. The achievability result obtained for MARCs applies directly to MABRCs as well.

    2) We derive two sets of necessary conditions for the achievability of  source-channel rates for MARCs (and MABRCs).

    3) We investigate MARCs and MABRCs subject to independent and identically distributed (i.i.d.) fading, for both phase fading and Rayleigh fading. We find
    conditions under which {\em informational source-channel separation (in the sense of \cite[Section I]{Tuncel:06}) is optimal for each channel model}.
    {\em This is the first time the optimality of separation is proven for some special case of MARCs and MABRCs}. Note that these models are not degraded in the sense of \cite{Sankar:09}.

    4) We derive two joint source-channel coding achievability schemes for MARCs and MABRCs for the source-channel rate $\kappa=1$. Both proposed schemes use a combination of
    Slepian-Wolf (SW) source coding \cite{SW:73} and joint source-channel coding implemented via the correlation preserving mapping (CPM) technique \cite{Cover:80}.
    In the first scheme CPM is used for encoding information to the relay and SW source coding combined with an independent channel code is used for encoding information to the destination. In the second scheme,
    SW source coding is used for encoding information to the relay and CPM is used for encoding information to the destination. These are the first {\em joint source-channel
    achievability schemes, proposed for a multiuser network with a relay, which take advantage of the CPM technique.}

\subsection*{Prior Work}

The MARC has been extensively studied from a channel coding perspective. Achievable rate regions for the MARC were derived in \cite{Kramer:2005}, \cite{Sankar:07} and \cite{KramerMandayam:04}. In \cite{Kramer:2005}
Kramer et al. derived an achievable rate region for the MARC with independent messages. The coding scheme employed in \cite{Kramer:2005} is based on decode-and-forward relaying, and uses regular encoding, successive
decoding at the relay, and backward decoding at the destination. In \cite{Sankar:07} it was shown that, in contrast to the classic relay channel, in a MARC different DF schemes yield different rate regions.
In particular, backward decoding can support a larger rate region than sliding window decoding. Another DF-based coding scheme which uses offset encoding, successive decoding at the relay, and sliding-window decoding
at the destination was presented in \cite{Sankar:07}. Outer bounds on the capacity region of MARCs were obtained in \cite{KramerMandayam:04}. More recently, capacity regions for two classes of MARCs were characterized
in \cite{Tandon:CISS:11}.

In \cite{ShamaiVerdu:95}, Shamai and Verd\'u considered the availability of correlated side information at the receiver in a point-to-point scenario, and showed that source-channel separation still holds. The availability
of correlated side information at the receiver enables transmitting the source reliably over a channel with a smaller capacity compared to the capacity needed in the absence of side information. In \cite{Cover:80} Cover et al. derived
finite-letter sufficient conditions for communicating discrete, arbitrarily correlated sources over a MAC, and showed the suboptimality of source-channel separation when transmitting correlated sources over a MAC. These sufficient conditions were later shown in \cite{Dueck:81} not to be necessary in general. The transmission technique
introduced by Cover et al. is called \emph{correlation preserving mapping (CPM)}. In the CPM technique the channel codewords are correlated with the source sequences, resulting in correlated channel inputs.
CPM is extended to source coding with side information over a MAC in \cite{Ahlswede:83} and to broadcast channels with correlated sources in \cite{HanCosta:87} (with a correction in \cite{KramerNair:09}).

%
%
%

In \cite{Tuncel:06} Tuncel distinguished between two types of source-channel separation. The first type, called {\em informational separation}, refers to classical separation in the Shannon sense. The second type,
called {\em operational separation}, refers to statistically independent source and channel codes, which are not necessarily the optimal codes for the underlying source or the channel, coupled with a joint decoder
at the destination. Tuncel also showed that when broadcasting a common source to multiple receivers, each with its own  correlated side information,
operational separation is optimal while informational separation is not.


In \cite{GunduzErkip:09} G\"und\"uz et al. obtained necessary and sufficient conditions for the optimality of informational separation in MACs with correlated sources and side information at the receiver.
The work \cite{GunduzErkip:09} also provided necessary and sufficient conditions for the optimality of operational separation for the compound MAC. Transmission of arbitrarily correlated sources over
interference channels (ICs) was studied in \cite{SalehiKurtas:93}, in which Salehi and Kurtas applied the CPM technique; however, when the sources are independent, the conditions derived in \cite{SalehiKurtas:93}
do not specialize to the Han and Kobayashi (HK) region, \cite{HK:81}, which is the largest known achievable rate region for ICs. Sufficient conditions based on the CPM technique, which specialize to the HK region were derived in \cite{LiuChen:2011}.
Transmission of independent sources over ICs with correlated receiver side information was studied in \cite{GunduzLiu:Rev}. The work \cite{GunduzLiu:Rev} showed that source-channel separation is optimal when each receiver has
access to side information correlated with its own desired source. When each receiver has access to side information correlated with the interfering transmitter's source, \cite{GunduzLiu:Rev} provided sufficient
conditions for reliable transmission based on a joint source-channel coding scheme which combines Han-Kobayashi superposition encoding and partial interference cancellation.


Lossless transmission over a relay channel with correlated side information was studied in \cite{ErkipGunduz:07}, \cite{Gunduz:12}, \cite{ElGamalCioffi:07} and \cite{Sefidgaran:ISIT:09}.
In \cite{ErkipGunduz:07} G\"und\"uz and Erkip developed a DF-based achievability scheme and showed that operational separation is optimal for physically degraded relay channels as well as for
cooperative relay-broadcast channels. The scheme of \cite{ErkipGunduz:07} was extended to multiple relay networks in \cite{Gunduz:12}.

Prior work on source transmission over fading channels is mostly limited to point-to-point channels (see \cite{Gunduz:IT:08} and references therein). In this work we consider two types of
fading models: phase fading and Rayleigh fading. Phase fading models apply to high-speed microwave communications where the oscillator's phase noise and the sampling clock jitter are the key impairments.
Phase fading is also the major impairment in communication systems that employ orthogonal frequency division multiplexing \cite{BarNess:04}. Additionally, phase fading can be used to model systems which employ
dithering to decorrelate signals \cite{Erez:06}.
For cooperative multi-user scenarios, phase-fading models have been considered for MARCs  \cite{Kramer:2005}, \cite{KramerMandayam:04}, \cite{Ron:2012}, for broadcast-relay channels (BRCs) \cite{Kramer:2005},
and for interference channels \cite{Goldsmith:07}.
Rayleigh fading models are very common in wireless communications and apply to mobile communications in the presence of multiple scatterers without line-of-sight \cite{Sklar:97}.
Rayleigh fading models have been considered for relay channels in \cite{Farhadi:08}, \cite{Farhadi:09} and \cite{Simeone:09}, and for MARCs in \cite{Ron:2012}.
The key similarity between the two fading models is the uniformly distributed phase of the fading process. The phase fading and the Rayleigh fading models differ in the behavior of the fading magnitude component,
 which is fixed for the former but varies following a Rayleigh distribution for the latter.

The rest of this paper is organized as follows: in Section \ref{sec:NotationModel} the model and notations are presented. In Section \ref{sec:sepBasedMARC} an achievable source-channel rate based on operational
separation is presented. In Section \ref{sec:outerBounds} necessary conditions on the achievable source-channel rates are derived. In Section \ref{sec:sepOpt} the optimality of separation for correlated sources
transmitted over fading Gaussian MARCs is studied, and in Section \ref{sec:JointAchiev} two achievable schemes based on joint source-channel coding are derived. Concluding remarks are provided in Section \ref{sec:conclusions}, followed by the appendices.

\section{Notations and Model} \label{sec:NotationModel}
In the following we denote the set of real numbers with $\setR$, and the set of complex numbers with $\setC$.
We denote random variables (RVs) with upper-case letters, e.g. $X$, $Y$, and their realizations with lower case letters, e.g. $x$, $y$. A discrete RV $X$ takes values in a set $\mX$. We use $|\mX|$ to denote the cardinality
of a finite, discrete set $\mX$, $p_X(x)$ to denote the probability mass function (p.m.f.) of a discrete RV $X$ over $\mX$, and $f_X(x)$ to denote the probability density function (p.d.f.) of a continuous RV $X$ on $\setC$.
For brevity we may omit the subscript $X$ when it is the uppercase version of the sample symbol $x$.
We use $p_{X|Y}(x|y)$ to denote the conditional distribution of $X$ given $Y$.
We denote vectors with boldface letters, e.g. $\xvec$, $\yvec$; the $i$'th element of a vector $\xvec$ is
denoted by $x_i$, and we use $\xvec_i^j$ where $i<j$, to denote $(x_i, x_{i+1},...,x_{j-1},x_j)$; $x^j$ is a short form notation for $x_1^j$, and unless specified otherwise, $\xvec \triangleq x^n$.
We denote the empty set with $\phi$, and the complement of the set $\mB$ by $\mB^c$.
We use $H(\cdot)$ to denote the entropy of a discrete RV, and $I(\cdot;\cdot)$ to denote the mutual information between two RVs, as defined in \cite[Ch. 2, Ch. 9]{cover-thomas:it-book}. 
We use $\styp(X)$ to denote the set of $\eps$-strongly typical sequences with respect to the distribution
$p_X(x)$ on $\mX$, as defined in \cite[Ch. 5.1]{YeungBook}; when referring to a typical set we may omit the RVs from the notation, when these variables are clear from the context.
We use $\mCN(a,\sigma^2)$ to denote a proper, circularly symmetric, complex
Gaussian distribution with mean $a$ and variance $\sigma^2$ \cite{NM:93}, and $\E\{ \cdot \}$ to denote stochastic expectation.
We use $X-Y-Z$ to denote a Markov chain formed by the RVs $X,Y,Z$ as defined in \cite[Ch. 2.8]{cover-thomas:it-book}, and $X \independent Y$ to denote that $X$ is statistically independent of $Y$.


\subsection{Problem Formulation} \label{subsec:MARC}
The MARC consists of two transmitters (sources), a receiver (destination) and a relay.
Transmitter $i$ has access to the source sequence $S_i^m$, for $i=1,2$.
The receiver is interested in the lossless reconstruction of the source sequences observed by the two transmitters.
The relay has access to side information $W_3^m$, and the receiver has access to side information $W^m$.
The objective of the relay is to help the receiver decode the source sequences.
For the MABRC the relay is also interested in a lossless reconstruction of the source sequences.
Figure \ref{fig:MABRCsideInfo} depicts the MABRC with side information setup. The MARC is obtained when the reconstruction at the relay is omitted.
\begin{figure}[t]
    \centering
    \scalebox{0.43}{\includegraphics{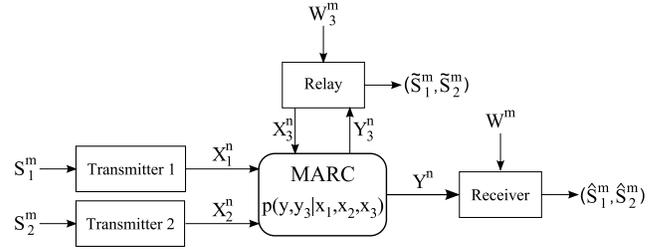}}
    \caption{Multiple-access broadcast relay channel with correlated side information. $(\tilde{S}^m_{1}, \tilde{S}^m_{2})$ are the reconstructions of $(S^m_{1}, S^m_{2})$ at the relay, and $(\hat{S}^m_{1}, \hat{S}^m_{2})$
    are the reconstructions at the destination.}
    \label{fig:MABRCsideInfo}
\end{figure}

The sources and the side information sequences, $\{ S_{1,k},S_{2,k},W_{k},W_{3,k} \}_{k=1}^{m}$, are
arbitrarily correlated according to a joint distribution $p(s_1,s_2,w,w_3)$ over a
finite alphabet $\mS_1 \times \mS_2 \times \mW \times \mW_3$, and independent across different sample indices $k$.
All nodes know this joint distribution.

For transmission, a discrete memoryless MARC with inputs $X_1, X_2, X_3$ over finite input alphabets $\mX_1,\mX_2,\mX_3$, and outputs $Y, Y_3$ over finite output alphabets $\mY,\mY_3$, is available.
The MARC is memoryless, that is,
\begin{align}
    & p(y_{k},y_{3,k}|y^{k-1}, y_{3,1}^{k-1}, x_{1,1}^k, x_{2,1}^k, x_{3,1}^k, s_{1,1}^m, s_{2,1}^m, w_{3,1}^m, w^m) \nonumber \\
		& \qquad = p(y_{k},y_{3,k}|x_{1,k},x_{2,k},x_{3,k}), \quad k=1,2,\dots,n.
\label{eq:MARCchanDist}
\end{align}

\begin{MyDefinition}
    \label{def:MABRCcodeDef}
    \textnormal{An $(m, n)$ source-channel code for the MABRC with correlated side information consists of two encoding functions,
    \begin{equation}
        f_i^{(m,n)} : \mS_i^m \mapsto \mX_i^n, \qquad i=1,2,
    \label{eq:MABRC_encFunc}
    \end{equation}
a set of causal encoding functions at the relay, $\{ f_{3,k}^{(m,n)} \}_{k=1}^n$, such that
    \begin{equation}
        x_{3,k} = f_{3,k}^{(m,n)}(y_{3,1}^{k-1},w_3^m), \qquad 1 \leq k \leq n,
    \label{eq:MABRC_relayEncFunc}
    \end{equation}
and two decoding functions
    \begin{subequations}
    \begin{align}
        g^{(m,n)} &: \mY^n \times \mW^m \mapsto \mS_1^m \times \mS_2^m, \\
        g_3^{(m,n)} &: \mY_3^n \times \mW_3^m \mapsto \mS_1^m \times \mS_2^m.
    \label{eq:MABRC_decFunc}
    \end{align}
    \end{subequations}
		}
\end{MyDefinition}
An $(m, n)$ source-channel code for the MARC is defined as in Definition \ref{def:MABRCcodeDef} with the exception that the decoding function $g_3^{(m,n)}$ does not exist.

\begin{MyDefinition}
    \label{def:MABRCpErr}
    \textnormal{Let $\hat{S}_i^m$ denote the reconstruction of $S_i^m$ at the receiver, and $\tilde{S}_{i}^m$ denote the reconstruction of $S_i^m$ at the relay, for $i=1,2$.
    The average probability of error, $P_e^{(m,n)}$, of an $(m,n)$ code for the MABRC is defined as
    \begin{align}
        P_e^{(m,n)} \triangleq \Pr & \Big\{ \big\{(\hat{S}_1^m,\hat{S}_2^m) \neq (S_1^m,S_2^m)\big\} \nonumber \\
				& \qquad \bigcup \big\{(\tilde{S}_{1}^m,\tilde{S}_{2}^m) \neq (S_1^m,S_2^m)\big\}  \Big\},
    \label{eq:MABRC_pErr}
    \end{align}}

    \noindent \textnormal{while for the MARC the average probability of error is defined as
    \begin{eqnarray}
        P_e^{(m,n)} & \triangleq & \Pr \Big\{(\hat{S}_1^m,\hat{S}_2^m) \neq (S_1^m,S_2^m) \Big\}
        .
    \label{eq:MARC_pErr}
    \end{eqnarray}}
\end{MyDefinition}

\begin{MyDefinition}
    \label{def:MABRCrate}
\textnormal{A source-channel rate $\kappa$ is said to be achievable for the MABRC if, for every $\epsilon > 0$, there exist positive integers $n_0, m_0$ such that for all $n>n_0, m>m_0,n \leq \kappa m  $ there exists an
$(m,n)$ code for which $P_e^{(m,n)} < \epsilon$.
    The same definition applies to the MARC.}
\end{MyDefinition}

\subsection{Fading Gaussian MARCs}
\label{sec:Fading Gaussian MARCs}
    The fading Gaussian MARC is depicted in Figure \ref{fig:FadeGaussMARC}. In fading Gaussian MARCs, the received signals at time $k$ at the receiver and at the relay are given by
    \begin{subequations}    \label{eq:FadeGaussMARC}
    \begin{eqnarray}
        Y_{k} &=& H_{11,k} X_{1,k} + H_{21,k} X_{2,k} + H_{31,k} X_{3,k} + Z_{k}, \\
        Y_{3,k} &=& H_{13,k} X_{1,k} + H_{23,k} X_{2,k} + Z_{3,k},
    \end{eqnarray}
    \end{subequations}

    \begin{figure}[t]
    \centering
    \scalebox{0.4}{\includegraphics{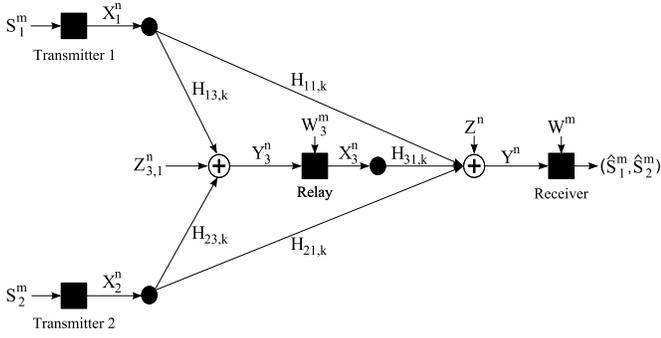}}
    \caption{Sources transmitted over fading Gaussian MARC with side information at the relay and destination.}
    \label{fig:FadeGaussMARC}
    \end{figure}

    \noindent for $k=1,\dots,n$, where $Z$ and $Z_3$ are independent of each other,
 i.i.d., circularly symmetric, complex Gaussian RVs, $\mCN(0,1)$.
 The channel input signals are subject to per-symbol average power constraints: $\E\{\left|X_i\right|^2 \} \leq P_i, i=1,2,3$.
 In the following it is assumed that the destination knows the instantaneous channel coefficients from the transmitters and the relay to itself, and the relay knows the instantaneous channel coefficients from
 both transmitters to itself.
This is referred to as receiver channel state information (Rx-CSI). Note that the destination does not have CSI on the links arriving at the relay, and that the relay does not have CSI on the links arriving at the destination.
 It is also assumed that the sources and the relay do not know the channel coefficients on their outgoing links (no transmitter CSI).
    We represent the CSI at the destination with $\tH_1 \triangleq \big( H_{11}, H_{21}, H_{31}\big)$, the CSI at the relay with $\tH_3 \triangleq \big(H_{13}, H_{23}\big)$, and define $\tH \triangleq \big\{H_{11}, H_{21}, H_{31}, H_{13}, H_{23} \big\}$.
    We consider two types of fading; phase fading and Rayleigh fading:
    \begin{enumerate}
    \item {\bf Phase fading channels}:
      The channel coefficients are characterized as $H_{li,k} = a_{li} e^{j\Theta_{li,k}}$, where $a_{li} \in \setR$ are constants representing the attenuation, and $\Theta_{li,k}$ are uniformly distributed over $[0,2\pi)$, i.i.d., and
      independent of each other and of the additive noises $Z_3$ and $Z$.

    \item {\bf Rayleigh fading channels}:
      The channel coefficients are characterized as $H_{li,k} = a_{li} U_{li,k}$, where $a_{li} \in \setR$ are constants representing the attenuation, and  $U_{li,k}$ are circularly symmetric, complex Gaussian RVs, $U_{li,k} \sim \CN(0,1)$, i.i.d., and independent of each
      other and of the additive noises $Z_3$ and $Z$. We define $\tU = \big\{U_{11}, U_{21}, U_{31}, U_{13}, U_{23} \big\}$.
    \end{enumerate}

    In both models the values of $a_{li}$ are fixed and known at all nodes. Observe that the magnitude of the phase-fading process is constant, $|H_{li,k}| = a_{li}$, but for Rayleigh fading the fading magnitude varies
between different time instances.

\section{An Achievable Source-Channel Rate Based on Operational Separation} \label{sec:sepBasedMARC}

In this section we derive an achievable source-channel rate for discrete memoryless (DM) MARCs and MABRCs using separate source and channel codes. The achievability is
established by using SW source coding, a channel coding scheme similar to the one detailed in \cite[Sections II, III]{Sankar:07}, and is based on DF relaying with irregular block Markov encoding,
successive decoding at the relay and backward decoding at the destination. The results are summarized in the following theorem:

\begin{theorem}
    \thmlabel{thm:separationCond}
    For DM MARCs and DM MABRCs with relay and receiver side information as defined in Section \ref{subsec:MARC}, source-channel rate $\kappa$ is achievable if,
    \begin{subequations} \label{bnd:sepBased}
    \begin{eqnarray}
        H(S_1|S_2,W_3) &<& \kappa I(X_1;Y_3|V_1, X_2, X_3) \label{bnd:rly_S1} \\
        H(S_2|S_1,W_3) &<& \kappa I(X_2;Y_3|V_2, X_1, X_3) \label{bnd:rly_S2} \\
        H(S_1,S_2|W_3) &<& \kappa I(X_1,X_2;Y_3|V_1, V_2, X_3) \label{bnd:rly_S1S2} \\
        H(S_1|S_2,W) &<& \kappa I(X_1,X_3;Y|V_2, X_2) \label{bnd:dst_S1} \\
        H(S_2|S_1,W) &<& \kappa I(X_2,X_3;Y|V_1, X_1) \label{bnd:dst_S2} \\
        H(S_1,S_2|W) &<& \kappa I(X_1,X_2,X_3;Y), \label{bnd:dst_S1S2}
    \end{eqnarray}
    \end{subequations}

    \noindent for some joint distribution $p(s_1,s_2,w_3,w,v_1,v_2,x_1,x_2,x_3)$ that factorizes as
    \begin{equation}
        p(s_1,s_2,w_3,w)p(v_1)p(x_1|v_1)p(v_2)p(x_2|v_2)p(x_3|v_1,v_2).
    \label{eq:SepJointDist}
    \end{equation}
\end{theorem}

\begin{IEEEproof}
    The proof is given in Appendix \ref{annex:proofSeparation}.
\end{IEEEproof}

\subsection{Discussion}
\ifthenelse{\boolean{SquizFlag}}{}{}

\begin{remark}
    In \Thmref{thm:separationCond}, equations \eqref{bnd:rly_S1}--\eqref{bnd:rly_S1S2} are constraints for reliable decoding at the relay, while equations \eqref{bnd:dst_S1}--\eqref{bnd:dst_S1S2} are reliable decoding constraints at the destination.
\end{remark}

\begin{remark}\label{rem:regIrregDiff}
    In regular encoding, the codebooks at the sources and at the relay have the same cardinality, see for example \cite{Sankar:07}. Now, note that the achievable source-channel rate of \Thmref{thm:separationCond} is established by using two different Slepian-Wolf coding schemes at different coding rates:
    one for the relay and one for the destination.
    The main benefit of different encoding rates is that it allows adapting to the different quality of side information at the relay and destination.
    Since the rates are different, such encoding cannot be realized with regular encoding and requires an irregular coding scheme for the channel code.

    Had we applied regular encoding, it would have led to the merger of some of the constraints in \eqref{bnd:sepBased}, in order to force the binning rates to the relay and destination to be equal.
    For example, \eqref{bnd:rly_S1} and \eqref{bnd:dst_S1} would be merged into the constraint
    \begin{align}
        &\max \big\{ H(S_1|S_2,W_3), H(S_1|S_2,W) \big\} \nonumber \\
				& \quad < \kappa \min \big\{I(X_1;Y_3|V_1, X_2, X_3), I(X_1,X_3;Y|V_2, X_2)\big\}. \nonumber
    \end{align}

    Hence, regular encoding puts extra constraints on the rates.
    Accordingly, we conclude that irregular encoding allows higher achievable source-channel rates than regular encoding.
    When the relay and destination have the same side information ($W\!=\!W_3$) then the irregular regular encoding schemes achieve the same source-channel rate.
    This can be observesd by setting $W=W_3$ in the above equation, and in \eqref{bnd:rly_S1} and \eqref{bnd:dst_S1}.

    Finally, consider {\em regular encoding} in the case of a MARC. Here, the relay is not required to recover the source sequences.
    Therefore, regular encoding requires the merger of the corresponding  right-hand sides (RHSs) of the constraints \eqref{bnd:rly_S1}--\eqref{bnd:dst_S1S2}. For example, \eqref{bnd:rly_S1} and \eqref{bnd:dst_S1} are merged into the following single constraint
        \begin{align}
            & H(S_1|S_2,W) \nonumber \\
						& \quad <  \kappa \min \big\{I(X_1;Y_3|V_1, X_2, X_3), I(X_1,X_3;Y|V_2, X_2)\big\}. \nonumber
        \end{align}

    \noindent This shows that regular encoding is more restrictive than irregular encoding for MARCs as well.
\end{remark}

\section{Necessary Conditions on the Achievable Source-Channel Rate for Discrete Memoryless MARCs and MABRCs} \label{sec:outerBounds}

In this section we derive necessary conditions for the achievability of a source-channel rate $\kappa$ for MARCs and for MABRCs with correlated sources and side information at the relay and at the destination.
The conditions for the MARC are summarized in the following theorem:


\begin{theorem}
    \thmlabel{thm:OuterGeneral}
    Consider the transmission of arbitrarily correlated sources $S_1$ and $S_2$ over the DM MARC with relay side information $W_3$ and receiver side information $W$.
    Any achievable source-channel rate $\kappa$ must satisfy the following constraints:
\begin{subequations} \label{bnd:outr_general_dst}
\begin{eqnarray}
        H(S_1|S_2,W) &\le& \kappa I(X_1,X_3;Y|X_2) \label{bnd:outr_general_dst_S1} \\
        H(S_2|S_1,W) &\le& \kappa I(X_2,X_3;Y|X_1) \label{bnd:outr_general_dst_S2} \\
        H(S_1,S_2|W) &\le& \kappa I(X_1,X_2,X_3;Y), \label{bnd:outr_general_dst_S1S2}
\label{eq:outerGeneral}
\end{eqnarray}
\end{subequations}

    \noindent for some input distribution $p(x_1,x_2,x_3)$, and the constraints
    \begin{subequations} \label{bnd:outr_V_dst}
\begin{eqnarray}
        H(S_1|S_2,W,W_3) &\le& \kappa I(X_1;Y,Y_3|X_2,V) \label{bnd:outr_V_dst_S1} \\
        H(S_2|S_1,W,W_3) &\le& \kappa I(X_2;Y,Y_3|X_1,V) \label{bnd:outr_V_dst_S2} \\
        H(S_1,S_2|W,W_3) &\le& \kappa I(X_1,X_2;Y,Y_3|V), \label{bnd:outr_V_dst_S1S2}
\label{eq:outerV}
\end{eqnarray}
\end{subequations}
    \noindent for some input distribution $p(v)(x_1,x_2|v)p(x_3|v)$, with $\left| \mV \right| \leq 4$.
\end{theorem}

\begin{IEEEproof}
    The proof is given below in Subsection \ref{subsec:proofOuter}.
\end{IEEEproof}

\begin{remark}
    The RHS of the constraints in \eqref{bnd:outr_V_dst} are similar to the broadcast bound\footnote{Here we use the common terminology for the classic relay channel in which the term $I(X, X_1;Y)$ is referred
    to as the MAC bound while the term $I(X;Y, Y_1|X_1)$ is called the broadcast bound \cite[Ch. 16]{KimElGamal:12}.}
 when assuming that all relay information is available at the destination.
\end{remark}

\begin{remark}
    Setting $\mX_2=\mS_2 = \phi$, constraints in \eqref{bnd:outr_general_dst} specialize to the converse of \cite[Thm. 3.1]{ErkipGunduz:07} for the relay~channel.
\end{remark}

\begin{theorem}
    \thmlabel{thm:OuterGeneralMABRC}

    Consider the transmission of arbitrarily correlated sources $S_1$ and $S_2$ over the DM MABRC with relay side information $W_3$ and receiver side information $W$.
    Any achievable source-channel rate $\kappa$ must satisfy the constraints \eqref{bnd:outr_general_dst} as well as the following constraints:
\begin{subequations} \label{bnd:outr_general_rly}
    \begin{eqnarray}
        H(S_1|S_2,W_3) &\le& \kappa I(X_1;Y_3|X_2, X_3) \label{bnd:outr_general_rly_S1} \\
        H(S_2|S_1,W_3) &\le& \kappa I(X_2;Y_3|X_1, X_3) \label{bnd:outr_general_rly_S2} \\
        H(S_1,S_2|W_3) &\le& \kappa I(X_1,X_2;Y_3| X_3) \label{bnd:outr_general_rly_S1S2},
\label{eq:outerRelayGeneral}
\end{eqnarray}
\end{subequations}

\noindent for some input distribution $p(x_1,x_2,x_3)$.

\end{theorem}

\ifthenelse{\boolean{SquizFlag}}{}{}

\ifthenelse{\boolean{SquizFlag}}{}{}


\subsection{Proof of Theorem \thmref{thm:OuterGeneral}} \label{subsec:proofOuter}

Let $P_e^{(m,n)} \rightarrow 0$ as $n,m \rightarrow \infty$, for a sequence of encoders and decoders $f_1^{(m,n)}, f_2^{(m,n)}, f_3^{(m,n)}, g^{(m,n)}$, such that $\kappa = n/m$ is fixed.
    By Fano's inequality, \cite[Thm. 2.11.1]{cover-thomas:it-book}, we have
    \begin{eqnarray}
        H(S_1^m, S_2^m| \hat{S}_{1}^m, \hat{S}_{2}^m) &\leq& 1 + m P_e^{(m,n)} \log \left| \mS_1 \times \mS_2 \right| \nonumber \\
        &\triangleq& m\delta(P_e^{(m,n)}),
    \label{eq:FanoDest}
    \end{eqnarray}

    \noindent where $\delta(x)$ is a non-negative function that approaches $\frac{1}{m}$ as $x \rightarrow 0$.
    Observe that
    \begin{align}
        H(S_1^m, S_2^m| \hat{S}_{1}^m, \hat{S}_{2}^m) & \stackrel{(a)}{\geq} H(S_1^m, S_2^m| Y^n, W^m, \hat{S}_{1}^m, \hat{S}_{2}^m) \nonumber \\
				& \stackrel{(b)}{\geq} H(S_1^m | Y^n, W^m, S_2^m) \label{eq:generalDestEntropyBasic},
    \end{align}

    \noindent where (a) follows from the fact that conditioning reduces entropy \cite[Thm. 2.6.5]{cover-thomas:it-book}; and (b) follows from the fact that
$(\hat{S}_{1}^m, \hat{S}_{2}^m)$ is a function of $(Y^n,W^m)$.

\subsubsection{Proof of constraints \eqref{bnd:outr_general_dst}} \label{subsubsec:proofOuterGeneral}
    Constraint \eqref{bnd:outr_general_dst_S1} is a consequence of the following chain of inequalities:
    \begin{align}
        & \sum_{k=1}^{n}{I(X_{1,k}, X_{3,k};Y_{k}| X_{2,k})} \nonumber \\
        & \mspace{15mu} \stackrel{(a)}{=} \sum_{k=1}^{n}{\Big[H(Y_{k}|X_{2,k})} - H(Y_{k}|S_1^m, S_2^m, W_3^m, W^m,  \nonumber \\
				& \mspace{220mu} X_{1,1}^{k}, X_{2,1}^{k}, X_{3,1}^{k}, Y_{3,1}^{k-1}, Y^{k-1})\Big] \nonumber \\
%
        & \mspace{15mu} \stackrel{(b)}{\geq} \sum_{k=1}^{n}{\Big[H(Y_{k}|S_2^m, W^m, Y^{k-1},X_{2,k})} - \nonumber \\
				& \mspace{140mu} H(Y_{k}|S_1^m, S_2^m, W_3^m, W^m, Y^{k-1})\Big] \nonumber \\
%
%
        & \mspace{15mu} \stackrel{(c)}{=} I(S_1^m, W_3^m; Y^n|S_2^m, W^m) \nonumber \\
%
%
        & \mspace{15mu} \stackrel{(d)}{\ge} H(S_1^m | S_2^m, W^m) - H(S_1^m |Y^n, S_2^m, W^m) \nonumber \\
%
        & \mspace{15mu} \stackrel{(e)}{\geq} mH(S_1 | S_2, W) - m\delta(P_e^{(m,n)})         \label{eq:generalDestSingleChain_11},
    \end{align}

    \noindent where (a) follows from the memoryless channel assumption (see \eqref{eq:MARCchanDist}) and the Markov relation $(S_1^m, S_2^m, W_3^m, W^m)-(X_{1,1}^{k}, X_{2,1}^{k}, X_{3,1}^{k}, Y_{3,1}^{k-1}, Y^{k-1})-Y_{k}$
    (see \cite{Massey:90}); (b) follows from the fact that conditioning reduces entropy;
(c) follows from the fact that $X_{2,k}$ is a deterministic function of $S_2^m$;
 (d) follows from the non-negativity of the mutual information; and (e) follows from the
memoryless sources and side information assumption, and from \eqref{eq:FanoDest}--\eqref{eq:generalDestEntropyBasic}.

    Following arguments similar to those that led to \eqref{eq:generalDestSingleChain_11} we can also show
    \begin{subequations}
    \begin{align}
        & \sum_{k=1}^{n}{I(X_{2,k}, X_{3,k};Y_{k} | X_{1,k})} \nonumber \\
				& \qquad \qquad \geq mH(S_2|S_1, W) - m\delta(P_e^{(m,n)}) \label{eq:generalDestSingleChain2} \\
        & \sum_{k=1}^{n}{I(X_{1,k},X_{2,k}, X_{3,k};Y_{k})} \nonumber \\
				& \qquad \qquad \geq mH(S_1, S_2| W) - m\delta(P_e^{(m,n)}).
    \label{eq:generalDestJointChain}
    \end{align}
    \end{subequations}

\noindent We now recall that the mutual information is concave in the set
of joint distributions $p(x_1,x_2,x_3)$, \cite[Thm. 2.7.4]{cover-thomas:it-book}. Thus, taking the limit as $m,n \rightarrow \infty$ and letting
$P_e^{(m,n)} \rightarrow 0$, \eqref{eq:generalDestSingleChain_11}, \eqref{eq:generalDestSingleChain2} and \eqref{eq:generalDestJointChain} result in the constraints in \eqref{bnd:outr_general_dst}.

\subsubsection{Proof of constraints \eqref{bnd:outr_V_dst}} \label{subsubsec:proofOuterV}
We begin by defining the following auxiliary RV:
\begin{equation}
    V_k \triangleq (Y_{3,1}^{k-1}, W_3^m), \quad k=1,2,\dots,n.
\label{eq:VRelayAux}
\end{equation}

    Constraint \eqref{bnd:outr_V_dst_S1} is a consequence of the following chain of inequalities:
    \begin{align}
        & \sum_{k=1}^{n}{I(X_{1,k};Y_{k}, Y_{3,k}| X_{2,k}, V_{k})} \nonumber \\
%
        & \mspace{15mu} \stackrel{(a)}{=} \sum_{k=1}^{n}{\Big[H(Y_{k}, Y_{3,k}|X_{2,k}, Y_{3,1}^{k-1}, W_3^m)} \nonumber \\
				& \mspace{80mu}-  H(Y_{k}, Y_{3,k}|X_{1}^{k}, X_{2}^{k}, X_{3}^{k}, Y_{3,1}^{k-1}, Y^{k-1}, W_3^m)\Big] \nonumber \\
%
%
        & \mspace{15mu} \stackrel{(b)}{=} \sum_{k=1}^{n}{\Big[H(Y_{k}, Y_{3,k}|X_{2,k}, Y_{3,1}^{k-1}, Y^{k-1},W_3^m, W^m, S_2^m)} \nonumber \\
        & \mspace{80mu} - H(Y_{k}, Y_{3,k}|X_{1}^{k}, X_{2}^{k}, X_{3}^{k}, \nonumber \\
				& \mspace{150mu} Y_{3,1}^{k-1}, Y^{k-1}, W_3^m, W^m, S_1^m, S_2^m)\Big] \nonumber \\
%
%
        & \mspace{15mu} \stackrel{(c)}{\geq} \sum_{k=1}^{n}{\Big[H(Y_{k}, Y_{3,k}| Y_{3,1}^{k-1}, Y^{k-1},W_3^m, W^m, S_2^m)} \nonumber \\
				& \mspace{80mu} - H(Y_{k}, Y_{3,k}|Y_{3,1}^{k-1}, Y^{k-1}, W_3^m, W^m, S_1^m, S_2^m)\Big] \nonumber \\
%
%
        & \mspace{15mu} = I(S_1^m;Y^{n}, Y_{3}^n|W_3^m, W^m, S_2^m) \nonumber \\
%
%
        & \mspace{15mu} \ge H(S_1^m|W_3^m, W^m, S_2^m) - H(S_1^m|Y^n, W_3^m, W^m, S_2^m) \nonumber \\
%
%
        & \mspace{15mu} \stackrel{(d)}{\geq} mH(S_1 | S_2, W, W_3) - m\delta(P_e^{(m,n)}) \label{eq:VDestSingleChain},
    \end{align}

    \noindent where (a) follows from \eqref{eq:VRelayAux}, the fact that $X_{3,1}^k$ is a deterministic function of $(Y_{3,1}^{k-1}, W_3^m)$, and the memoryless channel assumption, (see \eqref{eq:MARCchanDist});
    (b) follows from the fact that conditioning reduces entropy and causality, \cite{Massey:90};
(c) follows from the fact that $X_{2,k}$ is a deterministic function of $S_2^m$, and conditioning reduces entropy; (d) follows again from the fact that conditioning reduces entropy, the memoryless sources
and side information assumption, and \eqref{eq:FanoDest}--\eqref{eq:generalDestEntropyBasic}.

Following arguments similar to those that led to \eqref{eq:VDestSingleChain} we can also show that
    \begin{subequations}
    \begin{align}
        & \sum_{k=1}^{n}{I(X_{2,k};Y_{k}, Y_{3,k}| X_{1,k}, V_{k})} \nonumber \\
				& \qquad \qquad \geq mH(S_2 | S_1, W, W_3) - m\delta(P_e^{(m,n)}) \label{eq:VDestSingleChain2} \\
				& \sum_{k=1}^{n}{I(X_{1,k}, X_{2,k};Y_{k}, Y_{3,k}| V_{k})} \nonumber \\
				& \qquad \qquad \geq mH(S_1,S_2| W, W_3) - m\delta(P_e^{(m,n)}) \label{eq:VDestJointChain}.
    \end{align}
    \end{subequations}

    Next we introduce the time-sharing RV $Q$, independent of all other RVs, and we have $Q=k$ with probability $1/n, k \in \{1,2,\dots,n \}$. We can write
    \begin{align}
        & \hspace{-0.5cm} \frac{1}{n}\sum_{k=1}^{n}{I(X_{1,k};Y_{k}, Y_{3,k}| X_{2,k}, V_{k})} \nonumber \\
				& \qquad \qquad = I(X_{1,Q};Y_{Q}, Y_{3,Q}| X_{2,Q}, V_{Q})  \nonumber \\
				& \qquad \qquad = I(X_{1};Y, Y_{3}| X_{2}, V),
    \label{eq:VouterConcavity}
    \end{align}

    \noindent where $X_1 \triangleq X_{1,Q}$, $X_2 \triangleq X_{2,Q}$, $Y \triangleq Y_{Q}$, $Y_3 \triangleq Y_{3,Q}$ and $V \triangleq (V_Q, Q)$.
    Since $(X_{1,k}, X_{2,k})$ and $X_{3,k}$ are independent given $V_k = (Y_3^{k-1}, W_3^m)$, for $\bar{v} = (v,k)$ we have
    \begin{align}
        & \mspace{-10mu} \Pr \{ X_1=x_1,X_2=x_2,X_3=x_3|V=\bar{v} \}  \nonumber \\
				& \mspace{2mu}  = \Pr \{\mspace{-2mu} X_1 \mspace{-1mu}= \mspace{-1mu}x_1, X_2 \mspace{-1mu}= \mspace{-1mu}x_2|V \mspace{-1mu} = \mspace{-1mu} \bar{v} \} \mspace{-1mu} \Pr \{X_3 \mspace{-1mu}= \mspace{-1mu}x_3|V \mspace{-1mu}= \mspace{-1mu}\bar{v} \}.
    \label{eq:VouterDist}
    \end{align}

\noindent Hence, the probability distribution is of the form given in \Thmref{thm:OuterGeneral} for the constraints in \eqref{bnd:outr_V_dst}.
%
%
\noindent Finally, repeating the steps leading to \eqref{eq:VouterConcavity} for \eqref{eq:VDestSingleChain2} and \eqref{eq:VDestJointChain}, and taking the limit $m,n \rightarrow \infty$, leads to the constraints in \eqref{bnd:outr_V_dst}.


\section{Optimality of Source-Channel Separation for Fading Gaussian MARCs and MABRCs}
\label{sec:sepOpt}

In this section we study source-channel coding for fading Gaussian MARCs and MABRCs. We derive conditions for the optimality of source-channel separation
for the phase and Rayleigh fading models. We begin by considering phase fading Gaussian MARCs, defined in  \eqref{eq:FadeGaussMARC}. The result is stated in the following theorem:

    \begin{theorem} \thmlabel{thm:PhaseGaussMARC}
        Consider the transmission of arbitrarily correlated sources $S_1$ and $S_2$ over a phase fading Gaussian MARC with receiver side information $W$ and relay side information $W_3$. Let the channel inputs be subject to
        per-symbol power constraints specified by
        \begin{equation}
            \E\{\left|X_i\right|^2 \} \leq P_i, \qquad i=1,2,3,
        \label{eq:phasePowerConstraint}
        \end{equation}

        \noindent and the channel coefficients and power constraints $\{P_i\}_{i=1}^{3}$ satisfy
        \begin{subequations} \label{eq:phaseRelayDecConstraints}
        \begin{eqnarray}
            a_{11}^2 P_1 + a_{31}^2 P_3 &\leq& a_{13}^2 P_1 \label{eq:phaseRelayDecConstraints_R1} \\
            a_{21}^2 P_2 + a_{31}^2 P_3 &\leq& a_{23}^2 P_2 \label{eq:phaseRelayDecConstraints_R2} \\
            a_{11}^2 P_1 + a_{21}^2 P_2 + a_{31}^2 P_3 &\leq& a_{13}^2 P_1 + a_{23}^2 P_2.
        \label{eq:phaseRelayDecConstraints_R1R2}
        \end{eqnarray}
        \end{subequations}

%
        \noindent A source-channel rate $\kappa$ is achievable if
        \begin{subequations} \label{eq:phaseAchievConstraints}
        \begin{align}
            H(S_1|S_2,W) &< \kappa \log_2(1 + a_{11}^2 P_1 + a_{31}^2 P_3) \label{eq:phaseAchievConstraints_S1}    \\
            H(S_2|S_1,W) &< \kappa \log_2(1 + a_{21}^2 P_2 + a_{31}^2 P_3) \label{eq:phaseAchievConstraints_S2}    \\
            H(S_1,S_2|W) &< \kappa \log_2(1 + a_{11}^2 P_1 + a_{21}^2 P_2 + a_{31}^2 P_3). \label{eq:phaseAchievConstraints_S1S2}
        \end{align}
        \end{subequations}

        \noindent Conversely, if source-channel rate $\kappa$ is achievable, then conditions \eqref{eq:phaseAchievConstraints} are satisfied with $<$ replaced by $\leq$.
    \end{theorem}

    \begin{IEEEproof}
        The necessity part is proved in Subsection \ref{subsec:converseFadeProof} and sufficiency is shown in subsection \ref{subsec:achieveFadeProof}.
    \end{IEEEproof}

    \begin{remark} \label{rem:achiid}
        To achieve the source-channel rates $\kappa$ stated in \Thmref{thm:PhaseGaussMARC} we use channel inputs distributed according to $X_i \sim \mCN(0,P_i), i \in \{1,2,3 \}$, all mutually independent,
        and generate the codebooks in an i.i.d. manner. The relay employs the DF scheme.
    \end{remark}


    \begin{remark}
			Note that the phase fading MARC is not degraded in the sense of \cite{Sankar:09}, see also \cite[Remark 33]{Kramer:2005}.
    \end{remark}

    \begin{remark}
        The result of \Thmref{thm:PhaseGaussMARC}  relies on the assumptions of additive Gaussian noise, Rx-CSI, and i.i.d. fading coefficients such that the phases
        of the  fading coefficients are mutually independent, uniformly distributed, and independent of their magnitudes. These assumptions are essential for the result.
    \end{remark}

    \begin{remark}
        Observe that from the achievability result of \cite[Appendix A]{Ron:2012}, it follows that the optimal source code and channel code used in the proof of \Thmref{thm:PhaseGaussMARC} are separate and stand-alone.
        Thus, informational separation is optimal.
        We now provide an intuitive explanation for the optimality of separation for the current scenario:
        Note that when separate and stand-alone source and channel codes are used, the channel
        inputs of the two transmitters, $X_1$ and $X_2$, are be mutually independent, i.e.,
        $p(x_1,x_2) = p(x_1)p(x_2)$. This puts a restriction on the feasible joint distributions for generating the channel codebooks.
        Using a joint source-channel code allows generating channel inputs that are statically dependent
        on the source symbols. Since $S_1$ and $S_2$ are correlated this induces statistical dependence between the channel inputs $X_1$ and $X_2$. This, in turn, enlarges the set of
        feasible joint input distributions which can be realized for generating the channel codebooks; and therefore, the set of achievable transmission rates over the channel may increase.
        However, for fading Gaussian MARCs, due to the uniformly distributed phases of the channel coefficients, in the absence of Tx-CSI, the received signal components (from the sources and from the relay)  at the destination are uncorrelated.
        Therefore, there is no advantage, from the perspective of channel coding, in generating correlated channel inputs.  Coupled with the entropy maximization property of the
        Gaussian RVs, we conclude that the optimal channel inputs are mutually independent. From this discussion it follows that there is no benefit from joint source-channel coding, and source-channel separation is optimal.

    \end{remark}

    \begin{remark}
        There exist examples of channels which are not fading Gaussian channels, but satisfy the  rest of the assumptions detailed in Section \ref{sec:Fading Gaussian MARCs}, for which
        the DF-based sufficient conditions of \Thmref{thm:separationCond} are not optimal. One such example is the Gaussian relay channel with fixed channel coefficients, see also discussion in
        \cite[Section VII-B]{Kramer:2005}.
\end{remark}


    Next, we consider source transmission over Rayleigh fading MARCs.

        \begin{theorem} \thmlabel{thm:RayleighGaussMARC}
        Consider transmission of arbitrarily correlated sources $S_1$ and $S_2$ over a Rayleigh fading Gaussian MARC with receiver side information $W$ and relay side information $W_3$.
        Let the channel inputs be subject to per-symbol power constraints as in \eqref{eq:phasePowerConstraint}, and let the channel coefficients and the power constraints $\{P_i\}_{i=1}^3$ satisfy
        \begin{subequations} \label{eq:RayleighRelayDecConstraints}
%
				\begin{align}
			1 + a_{11}^2 P_{1} + a_{31}^2P_{3} & \leq \frac{a_{13}^2 P_1} {e^{ \frac{ 1}{a_{13}^2 P_1}} E_1\left( \frac{1 }{ a_{13}^2  P_{1} }\right)} \\
			1 + a_{21}^2 P_{2} + a_{31}^2P_{3} & \leq \frac{a_{23}^2 P_2} {e^{ \frac{ 1}{a_{23}^2 P_2}} E_1\left( \frac{1 }{ a_{23}^2  P_{2} }\right)} \\
			1 + a_{11}^2 P_{1} + a_{21}^2 P_{2} + a_{31}^2P_{3} & \leq \nonumber
		\end{align}
		
		\vspace{-0.63cm}
		
		\begin{align}
			\mspace{50mu} \frac{a_{23}^2P_2 - a_{13}^2P_1}{\left( e^{ \frac{1}{a_{23}^2P_2}} E_1\left(\frac{1}{a_{23}^2P_2}\right) - e^{\frac{1}{a_{13}^2P_1}}E_1\left(\frac{1}{a_{13}^2P_1}\right) \right)},
		\end{align}
        \end{subequations}

        \noindent where $E_1(x) \triangleq \int_{q=x}^{\infty}{\frac{1}{q}e^{-q}dq}$, see \cite[Eqn. (5.1.1)]{Abramowitz}. A source-channel rate $\kappa$ is achievable if
        \begin{subequations} \label{eq:RayleighAchievConstraints}
        \begin{align}
			\mspace{-10mu} H(S_1|S_2,W) &< \kappa \E_{\tU}\big\{\log_2(1 + a_{11}^2|U_{11}|^2P_1 \nonumber \\ 		
				& \mspace{171mu}	+ a_{31}^2|U_{31}|^2P_3) \big\} \label{eq:RayleighAchievConstraints_S1}	\\
			\mspace{-10mu} H(S_2|S_1,W) &< \kappa \E_{\tU}\big\{\log_2(1 + a_{21}^2|U_{21}|^2P_2 \nonumber \\
			 & \mspace{171mu}	 + a_{31}^2|U_{31}|^2P_3) \big\} \label{eq:RayleighAchievConstraints_S2}	\\
			\mspace{-10mu} H(S_1,S_2|W) &< \kappa \E_{\tU}\big\{\log_2(1 + a_{11}^2|U_{11}|^2P_1 + \nonumber \\
			& \mspace{80mu} a_{21}^2|U_{21}|^2P_2  + a_{31}^2|U_{31}|^2P_3) \big\}, \label{eq:RayleighAchievConstraints_S1S2}
		\end{align}
        \end{subequations}

        \noindent Conversely, if source-channel rate $\kappa$ is achievable, then conditions \eqref{eq:RayleighAchievConstraints} are satisfied with $<$ replaced by $\leq$.
    \end{theorem}

    \begin{IEEEproof}
        The proof uses \cite[Corollary 1]{Ron:2012} and follows similar arguments to those in the proof of \Thmref{thm:PhaseGaussMARC}.
    \end{IEEEproof}

    \begin{remark}
        The source-channel rate $\kappa$ in \Thmref{thm:RayleighGaussMARC} is achieved by using $X_i \sim \mCN(0,P_i), i \in \{1,2,3 \}$, all i.i.d. and independent of each other, and applying DF at the relay.
    \end{remark}

\subsection{Proof of Theorem \thmref{thm:PhaseGaussMARC}}   \label{subsec:FadeProof}

\subsubsection{Necessity Proof of Theorem \thmref{thm:PhaseGaussMARC}} \label{subsec:converseFadeProof}

Consider the necessary conditions of \Thmref{thm:OuterGeneral}.
We first note that the phase fading MARC model specified in Section \ref{sec:Fading Gaussian MARCs} exactly belongs to the class of fading relay channels
with Rx-CSI\footnote{Rx-CSI is incorporated into \Thmref{thm:OuterGeneral} by replacing $Y$ with
$(Y,\tH_1)$ in Eqns. \eqref{bnd:outr_general_dst}, and $(Y,Y_3)$ with $(Y,Y_3,\tH)$ in Eqns. \eqref{bnd:outr_V_dst}, and then by using the
fact that due to the absence of Tx-CSI, $(\tH_1,\tH) \independent (X_1,X_2,X_3)$, see \cite[Eq. (50)]{Kramer:2005}.}
stated in \cite[Thm. 8]{Kramer:2005}. Thus, from \cite[Thm. 8]{Kramer:2005} it follows that for phase fading MARCs with Rx-CSI, the mutual information
expressions on the RHS of \eqref{bnd:outr_general_dst} are simultaneously maximized by $X_1,X_2,X_3$ mutually independent, zero-mean
complex Gaussian RVs, $X_i \sim \CN(0,P_i), i = 1,2,3$.
Applying this input p.d.f. to \eqref{bnd:outr_general_dst} yields the expressions  in \eqref{eq:phaseAchievConstraints}.
Therefore, for phase fading MARCs, when conditions \eqref{eq:phaseRelayDecConstraints} hold, the conditions in \eqref{eq:phaseAchievConstraints}
coincide with the necessary conditions of Thm. \thmref{thm:OuterGeneral}, after replacing
 $``<"$ with  $``\leq"$.

\subsubsection{Sufficiency Proof of Theorem \thmref{thm:PhaseGaussMARC}} \label{subsec:achieveFadeProof} $ $

    \textbullet \hspace{0.04cm} \textsl{Codebook construction:} 
For $i=1,2$, assign every $\svec_i \in \mS_i^m$~to one of $2^{mR_i}$ bins independently according to a uniform
distribution over $\msgCal_{i} \triangleq \{1,2,\dots,2^{mR_i}\}$. Denote these assignments by $f_i$.
Set $\chR_i = \frac{1}{\kappa} R_i$, $i=1,2$, and let $X_k \sim \CN(0,P_k)$, $k=1,2,3$, all mutually independent. Construct
a channel code based on DF with rates $\chR_1$ and $\chR_2$, and with blocklength $n$,  as detailed in \cite[Appendix A]{Ron:2012}.

\textbullet \hspace{0.04cm} \textsl{Encoding:} Consider sequences of length $Bm$, $s_i^{Bm} \in \mS^{Bm}_i, i=1,2$, $w^{Bm} \in \mW^{Bm}$.
Partition each sequence into $B$ length-$m$ subsequences, $\svec_{i,b}$, $i=1,2$, and $\wvec_b$, $b=1,2,\dots,B$.
A total of $Bm$ source samples are transmitted over $B+1$ blocks of $n$ channel symbols each. Setting $n=\kappa m$,
and increasing $B$ we obtain a source-channel rate $(B+1)n/Bm \rightarrow n/m = \kappa$ as $B \rightarrow \infty$.

\noindent At block $b, b=1,2,\dots,B$, source terminal $i, i=1,2$, observes $\svec_{i,b}$ and finds its corresponding
bin index $\msg_{i,b} \in \msgCal_{i}$. Each transmitter sends its corresponding bin index using the
channel code described in \cite[Appendix A]{Ron:2012}.
Assume that at time $b$ the relay knows 
$(\msg_{1,b-1},\msg_{2,b-1})$. The relay sends these bin indices using the encoding scheme described in \cite[Appendix A]{Ron:2012}.

\textbullet \hspace{0.04cm} \textsl{Decoding and error probability analysis:}
We apply the decoding rule of \cite[Eqn. (A2)]{Ron:2012}. From
the error probability analysis in \cite[Appendix A]{Ron:2012}, it follows that, when the channel coefficients and the channel input power constraints satisfy the
conditions in \eqref{eq:phaseRelayDecConstraints}, the RHSs of the constraints in \eqref{eq:phaseAchievConstraints} characterize the ergodic capacity
region (in the sense of \cite[Eq. (51)]{Kramer:2005}) of the phase fading Gaussian MARC (see \cite[Thm. 9]{Kramer:2005}, \cite[Appendix A]{Ron:2012}).
Hence, when consitions \eqref{eq:phaseRelayDecConstraints} are satisfied, the transmitted bin indices $\{u_{1,b},u_{2,b} \}_{b=1}^B$ can be
reliably decoded at the destination as long as
\begin{subequations} \label{eq:phaseAchieveChanRates}
\begin{align}
    R_1 & <  \kappa \log_2(1 + a_{11}^2 P_1 + a_{31}^2 P_3) \\
    R_2 & <  \kappa \log_2(1 + a_{21}^2 P_2 + a_{31}^2 P_3) \\
    R_1 + R_2 & < \kappa \log_2(1 + a_{11}^2 P_1 + a_{21}^2 P_2 + a_{31}^2 P_3).
\end{align}
\end{subequations}

\noindent \textit{Decoding the sources at the destination:} The decoded bin indices, denoted
$(\tilde{\msg}_{1,b},\tilde{\msg}_{2,b}), b=1,2,\dots,B$, are given to
the source decoder at the destination. Using the bin indices $(\tilde{\msg}_{1,b},\tilde{\msg}_{2,b})$
and the side information $\wvec_{b}$, the source decoder at the destination estimates $(\svec_{1,b}, \svec_{2,b})$
by looking for a unique pair of sequences $(\tilde{\svec}_{1},\tilde{\svec}_{2}) \in S_1^m \times S_2^m$
that satisfies $f_1(\tsvec_{1})= \tilde{\msg}_{1,b}, f_2(\tsvec_{2})= \tilde{\msg}_{2,b}$ and
$(\tsvec_{1},\tsvec_{2},\wvec_{b}) \in \stypm(S_1,S_2,W)$.
From the Slepian-Wolf theorem \cite[Thm 14.4.1]{cover-thomas:it-book}, $(\svec_{1,b},\svec_{2,b})$
can be reliably decoded at the destination if
\begin{subequations} \label{eq:phaseAchieveDestSourceRates}
\begin{eqnarray}
    H(S_1|S_2,W) &\le& R_1 \\
    H(S_2|S_1,W) &\le& R_2 \\
    H(S_1,S_2|W) &\le& R_1 + R_2.
\end{eqnarray}
\end{subequations}

 \noindent   Combining conditions \eqref{eq:phaseAchieveChanRates} and \eqref{eq:phaseAchieveDestSourceRates}
yields \eqref{eq:phaseAchievConstraints}, and completes the achievability proof.


\begin{remark}
    Note that in the sufficiency proof in Section \ref{subsec:achieveFadeProof} we used the code construction and the decoding procedure of \cite[Appendix A]{Ron:2012}, which are designed specifically for fading MARCs.
    The reason we did not use the result of \Thmref{thm:separationCond} is that for the channel inputs to be mutually independent, we must set $V_1 = V_2 = \phi$ in \Thmref{thm:separationCond}.
    But, with such an assignment, the decoding rule of the channel code at the destination given by Eqn. \eqref{eq:DestChanDecType} does not apply, as this rule decodes the information carried by the
    {\em auxiliary RVs}.
    For the same reason we did not simply cite \cite[Thm. 9]{Kramer:2005} for the channel coding part  of the sufficiency proof of Thm. \thmref{thm:PhaseGaussMARC}.
    We conclude that a specialized channel code must be constructed for fading channels.
    The issue of channel coding for fading MARCs has already been addressed in \cite{Ron:2012}, and we refer to \cite{Ron:2012} for a detailed discussion.
\end{remark}

\subsection{Fading MABRCs}

Optimality of informational separation can also be established for MABRCs by
using the results for MARCs  with three additional constraints. The result is stated in the following theorem:
\begin{theorem} \thmlabel{thm:PhaseGaussMABRC}
    For phase fading MABRCs for which the conditions in \eqref{eq:phaseRelayDecConstraints} hold together with
    \begin{subequations} \label{eq:phaseEntropyConstraints}
    \begin{eqnarray}
        H(S_1|S_2,W_3) &\leq&  H(S_1|S_2,W) \\
        H(S_2|S_1,W_3) &\leq&  H(S_2|S_1,W) \\
        H(S_1,S_2|W_3) &\leq&  H(S_1,S_2|W),
    \end{eqnarray}
    \end{subequations}

  \noindent a source-channel rate $\kappa$ is achievable if conditions \eqref{eq:phaseAchievConstraints} are satisfied.
  Conversely, if a source-channel rate $\kappa$ is achievable, then conditions \eqref{eq:phaseAchievConstraints} are satisfied with $<$ replaced by $\leq$.
    The same statement holds for Rayleigh fading MABRCs with \eqref{eq:RayleighRelayDecConstraints} replacing \eqref{eq:phaseRelayDecConstraints} and \eqref{eq:RayleighAchievConstraints} replacing \eqref{eq:phaseAchievConstraints}.
\end{theorem}

\begin{IEEEproof}
    The sufficiency proof of \Thmref{thm:PhaseGaussMABRC} differs from the sufficiency proof of \Thmref{thm:PhaseGaussMARC} only due to decoding requirement of the source sequences at the relay.
    Conditions \eqref{eq:phaseRelayDecConstraints} imply that  reliable decoding of the channel code at the destination implies reliable decoding of the channel code at the relay.
    Conditions \eqref{eq:phaseEntropyConstraints} imply that  the relay achievable source rate region contains the destination achievable source rate region, and therefore,
    reliable decoding of the source code at the destination implies reliable decoding of the source code at the relay.
    Hence, if conditions \eqref{eq:phaseRelayDecConstraints}, \eqref{eq:phaseAchievConstraints}, and \eqref{eq:phaseEntropyConstraints} hold, $(\svec_{1,b},\svec_{2,b})$, $b=1,2,...,B$,
    can be reliably decoded at both the relay and the destination.
    Necessity of \eqref{eq:phaseAchievConstraints} follows from the necessary conditions of \Thmref{thm:OuterGeneralMABRC}, and by following similar arguments to the necessity proof of \Thmref{thm:PhaseGaussMARC}.

    The extension to Rayleigh fading is similar to the one done for MARCs (from \Thmref{thm:PhaseGaussMARC} to \Thmref{thm:RayleighGaussMARC}).
\end{IEEEproof}

\begin{remark}
    Conditions \eqref{eq:phaseEntropyConstraints} imply that for the scenario described in \Thmref{thm:PhaseGaussMARC}
regular and irregular encoding yield the same source-channel achievable rates (see Remark \ref{rem:regIrregDiff}); hence, the channel code construction of \cite[Appendix A]{Ron:2012} can be used
without any rate loss.
\end{remark}

\section{Joint Source-Channel Achievable Rates for Discrete Memoryless MARCs and MABRCs} \label{sec:JointAchiev}

In this section we derive two sets of sufficient conditions for the achievability of source-channel rate $\kappa=1$ for DM MARCs and MABRCs with correlated sources and side information.
Both achievability schemes are established by using a combination of SW source coding, the CPM technique, and a DF scheme with successive decoding at the relay and backward decoding at the destination.
The techniques differ in the way the source codes are combined.
In the first scheme (\Thmref{thm:jointCond_CP}), SW source coding is used for encoding information to the destination and CPM is used for
encoding information to the relay. In the second scheme (\Thmref{thm:jointCondFlip_CP}), CPM is used for encoding information to the destination while SW
source coding is used for encoding information to the relay.

\ifthenelse{\boolean{SquizFlag}}{}{}

\ifthenelse{\boolean{SquizFlag}}{}{}

\subsection{Joint Source-Channel Coding for MARCs and MABRCs}
\Thmref{thm:jointCond_CP} and \Thmref{thm:jointCondFlip_CP} below present two new sets of sufficient conditions for the achievability of source-channel rate $\kappa=1$, obtained by combining SW source coding and CPM.
For the sources $S_1$ and $S_2$ we define common information in the sense of G\'acs, K\"orner \cite{GacsKorner:73} and Witsenhausen \cite{Witsenhausen:75}, as $T \triangleq h_1(S_1) = h_2(S_2)$, where $h_1$ and $h_2$ are
deterministic functions. We now state the theorems:

\begin{theorem}
    \thmlabel{thm:jointCond_CP}
    For DM MARCs and MABRCs with relay and receiver side information as defined in Section \ref{subsec:MARC}, and source pair $(S_1,S_2)$ with common part $T \triangleq h_1(S_1) = h_2(S_2)$,
    a source-channel rate $\kappa=1$ is achievable if,
    \begin{subequations} \label{bnd:Joint_CP}
    \begin{align}
        \mspace{-13mu} H(S_1|S_2,W_3) &< I(X_1;Y_3|S_2, V_1, X_2, X_3, W_3,Q) \label{bnd:Joint_rly_S1_CP} \\
        \mspace{-13mu} H(S_2|S_1,W_3) &< I(X_2;Y_3|S_1, V_2, X_1, X_3, W_3, Q) \label{bnd:Joint_rly_S2_CP} \\
        \mspace{-13mu} H(S_1,S_2|W_3,T) &< I(X_1, \mspace{-2mu} X_2;Y_3|V_1, V_2, X_3, W_3, T, Q) \label{bnd:Joint_rly_S1S2_CP} \\
        \mspace{-13mu} H(S_1,S_2|W_3) &< I(X_1,\mspace{-2mu} X_2;Y_3|V_1, V_2, X_3, W_3) \label{bnd:Joint_rly_S1S2_2_CP} \\
        \mspace{-13mu} H(S_1|S_2,W) &< I(X_1,\mspace{-2mu} X_3;Y|S_1, V_2, X_2, Q) \label{bnd:Joint_dst_S1_CP} \\
        \mspace{-13mu} H(S_2|S_1,W) &< I(X_2,\mspace{-2mu} X_3;Y|S_2, V_1, X_1, Q) \label{bnd:Joint_dst_S2_CP} \\
        \mspace{-13mu} H(S_1,S_2|W) &< I(X_1,\mspace{-2mu} X_2,\mspace{-2mu} X_3;Y|S_1,S_2, Q), \label{bnd:Joint_dst_S1S2_CP}
    \end{align}
    \end{subequations}

    \noindent for some joint distribution that factorizes as
    \begin{align}
        & p(s_1,s_2,w_3,w)p(q)p(v_1)p(x_1|s_1,v_1,q) \times \nonumber \\
				& \quad p(v_2)p(x_2|s_2,v_2,q)p(x_3|v_1,v_2)p(y_3,y|x_1,x_2,x_3).
    \label{eq:JntJointDist_CP}
    \end{align}

\end{theorem}	

\begin{IEEEproof}
    The proof is given in Appendix \ref{annex:jointProof_CP}.
\end{IEEEproof}

\begin{theorem}		\thmlabel{thm:jointCondFlip_CP}
    For DM MARCs and MABRCs with relay and receiver side information as defined in Section \ref{subsec:MARC}, and source pair $(S_1,S_2)$ with
    common part $T \triangleq h_1(S_1) = h_2(S_2)$, a source-channel rate $\kappa=1$ is achievable if,
    \begin{subequations} \label{bnd:JointFlip_CP}
    \begin{align}
        H(S_1|S_2,W_3) &< I(X_1;Y_3|S_1, X_2, X_3,Q) \label{bnd:JointFlip_rly_S1_CP} \\
        H(S_2|S_1,W_3) &< I(X_2;Y_3|S_2, X_1, X_3,Q) \label{bnd:JointFlip_rly_S2_CP} \\
        H(S_1,S_2|W_3) &< I(X_1,X_2;Y_3|S_1, S_2, X_3,Q) \label{bnd:JointFlip_rly_S1S2_CP} \\
        H(S_1|S_2,W) &< I(X_1,X_3;Y|S_2, X_2, W,Q) \label{bnd:JointFlip_dst_S1_CP} \\
        H(S_2|S_1,W) &< I(X_2,X_3;Y|S_1, X_1, W,Q) \label{bnd:JointFlip_dst_S2_CP} \\
        H(S_1,S_2|W,T) &< I(X_1,X_2,X_3;Y| W,T,Q) \label{bnd:JointFlip_dst_S1S2_CP} \\
        H(S_1,S_2|W) &< I(X_1,X_2,X_3;Y|W), \label{bnd:JointFlip_dst_S1S2_2_CP}
    \end{align}
    \end{subequations}

    \noindent for some joint distribution that factorizes as
    \begin{align}
        & p(s_1,s_2,w_3,w)p(q)p(x_1|s_1,q) \times \nonumber \\
				& \qquad \quad p(x_2|s_2,q)p(x_3|s_1,s_2,q)p(y_3,y|x_1,x_2,x_3).
    \label{eq:JntFlipJointDist_CP}
    \end{align}

\end{theorem}	

\begin{IEEEproof}
    The proof is given in Appendix \ref{annex:jointProofFlip_CP}.
\end{IEEEproof}

\ifthenelse{\boolean{IncludeJntOldProofs}}{}{}

\subsection{Discussion} \label{subsec:jointDiscussion}
Figure \ref{fig:Dependency} illustrates the Markov chains for the joint distributions considered in \Thmref{thm:jointCond_CP} and \Thmref{thm:jointCondFlip_CP}.
\begin{figure}[t]
\begin{center}
\subfloat[]{\scalebox{0.31}{\includegraphics{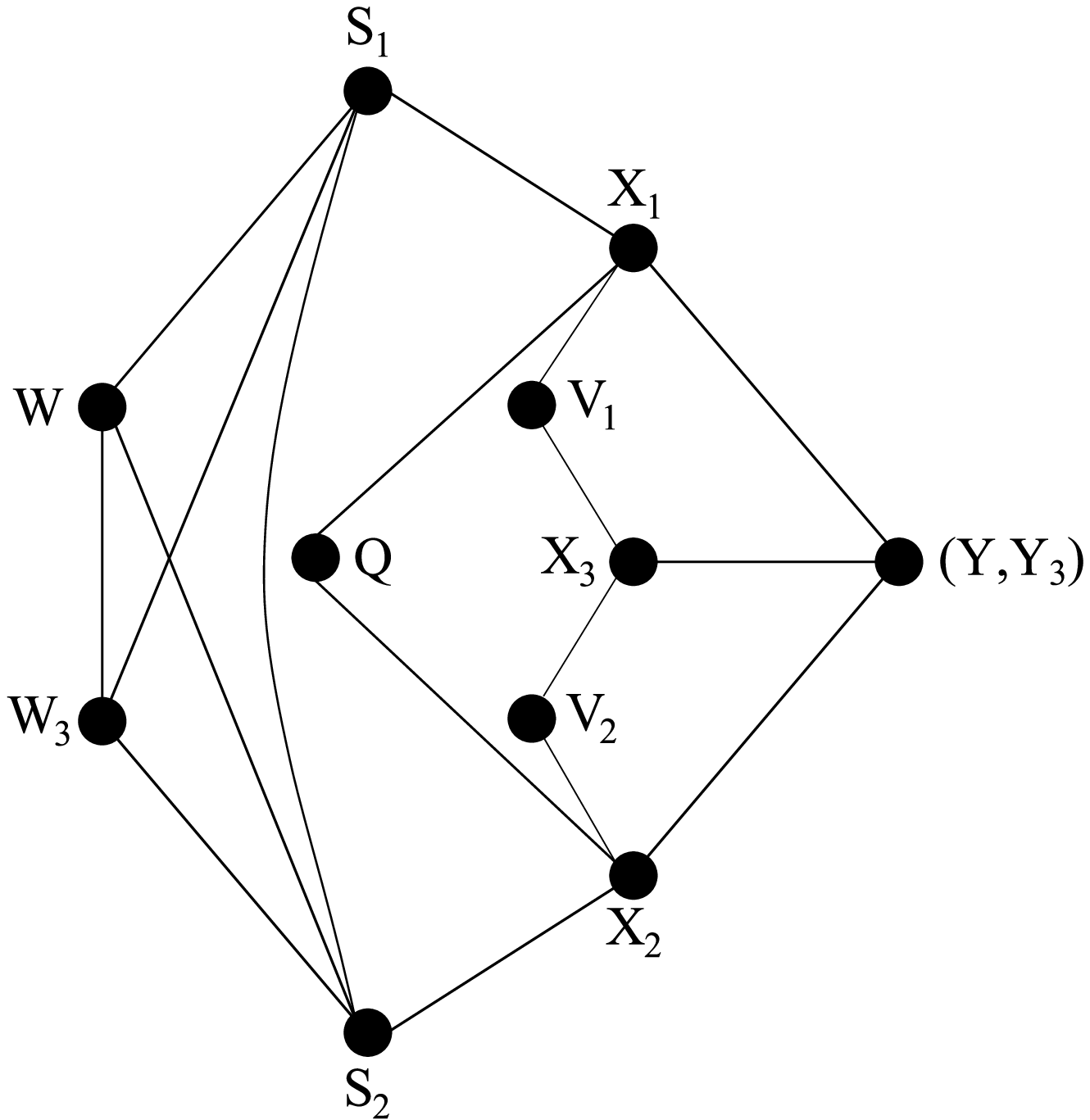}} \label{fig:JointDep}}
\subfloat[]{\scalebox{0.31}{\includegraphics{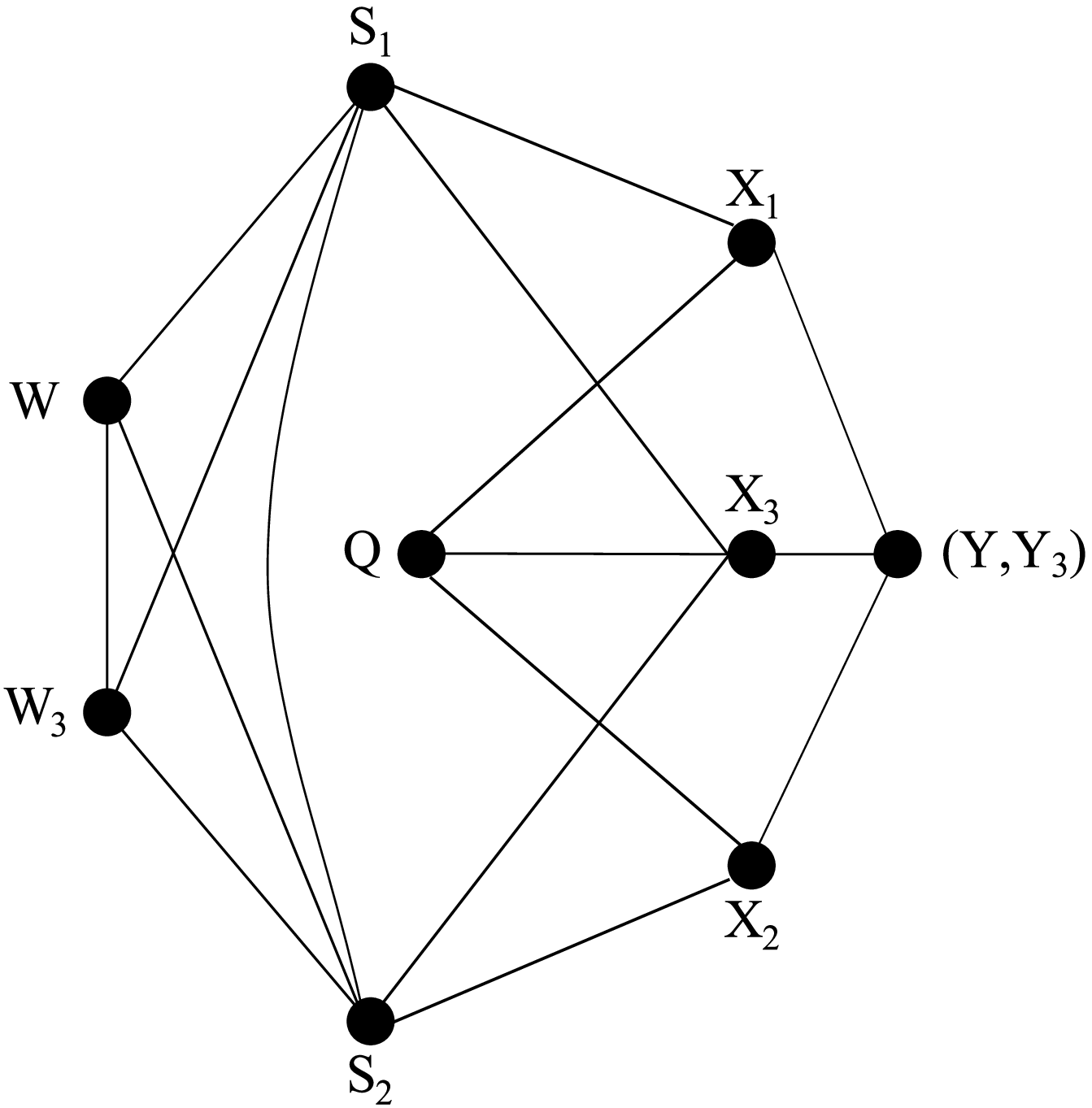}} \label{fig:JointFlipDep}}
\caption{(a) Diagram of the Markov chain for the joint distribution considered in \eqref{eq:JntJointDist_CP}; (b) Diagram of the Markov chain for the joint distribution considered in \eqref{eq:JntFlipJointDist_CP}.}
\label{fig:Dependency}
\end{center}
\end{figure}

\begin{remark}
    \label{rem:conds}
    Conditions \eqref{bnd:Joint_rly_S1_CP}--\eqref{bnd:Joint_rly_S1S2_2_CP} in \Thmref{thm:jointCond_CP} and  conditions \eqref{bnd:JointFlip_rly_S1_CP}--\eqref{bnd:JointFlip_rly_S1S2_CP}
    in \Thmref{thm:jointCondFlip_CP} are constraints for decoding at the relay, while conditions \eqref{bnd:Joint_dst_S1_CP}--\eqref{bnd:Joint_dst_S1S2_CP} and \eqref{bnd:JointFlip_dst_S1_CP}--\eqref{bnd:JointFlip_dst_S1S2_2_CP}
    are decoding constraints at the destination.
\end{remark}

\begin{remark}
    Each mutual information expression on the RHS of the constraints in \Thmref{thm:jointCond_CP} and \Thmref{thm:jointCondFlip_CP} represents the rate of one of two encoding types: either source-channel encoding via CPM
    or SW encoding.
	Consider \Thmref{thm:jointCond_CP}: Here, $V_1$ and $V_2$ represent the binning information for $S_1$ and $S_2$, respectively. Observe that the left-hand side (LHS) of condition
    \eqref{bnd:Joint_rly_S1_CP} is the entropy of $S_1$ when $(S_2,W_3)$ are known.
    On the RHS of \eqref{bnd:Joint_rly_S1_CP}, as $S_2$, $V_1$, $X_2$,  $X_3$, $W_3$  and  $Q$ are given, the mutual information expression
    $I(X_1;Y_3|S_2, V_1, X_2, X_3, W_3, Q)$ represents the available rate that can be used for encoding information on the {\em source} $S_1$, in {\em excess of the bin index} represented by $V_1$.
The LHS of condition \eqref{bnd:Joint_dst_S1_CP} is the entropy of $S_1$ when $(S_2,W)$ are known. The RHS of condition \eqref{bnd:Joint_dst_S1_CP} expresses the amount of {\em binning information} that can be reliably
transmitted cooperatively by transmitter 1 and the relay to the destination.
 This can be seen by rewriting the mutual information expression in \eqref{bnd:Joint_dst_S1_CP} as $I(X_1,X_3;Y|S_1, V_2, X_2, Q) = I(X_1,X_3;Y|S_1, S_2, X_2, V_2, W, Q)$. As $S_1$ is given, this expression represents the rate at which the {\em bin index} of source $S_1$ can be transmitted to the destination in excess of the source-channel rate for encoding $S_1$ (see Appendix \ref{annex:jointProof_CP}).
Therefore, each mutual information expression in \eqref{bnd:Joint_rly_S1_CP} and \eqref{bnd:Joint_dst_S1_CP} represents \emph{different} types of information sent by the source:
either source-channel codeword to the relay as in \eqref{bnd:Joint_rly_S1_CP}; or bin index to the destination as in \eqref{bnd:Joint_dst_S1_CP}.
This difference is because SW source coding is used for encoding information to the destination while CPM is used for encoding information to the relay.

  Similarly, consider the RHS of \eqref{bnd:JointFlip_rly_S1_CP} in \Thmref{thm:jointCondFlip_CP}. The mutual information expression $I(X_1;Y_3|S_1, X_2, X_3,Q) = I(X_1;Y_3|S_1, S_2, X_2, X_3, W_3, Q)$
  represents the rate that can be used for encoding the {\em bin index} of source $S_1$ to the relay (see Appendix \ref{annex:jointProofFlip_CP}), since $S_1$ is given.
In contrast, the mutual information expression $I(X_1,X_3;Y|S_2, X_2,W, Q)$ on the RHS of \eqref{bnd:JointFlip_dst_S1_CP} represents the available rate that can be used for cooperative source-channel
encoding of the {\em source} $S_1$ to the destination. This follows as $S_2$, $X_2$, $W$  and $Q$ are given.

\end{remark}

\begin{remark} \label{rem:tradeoff}
    \Thmref{thm:jointCond_CP} and \Thmref{thm:jointCondFlip_CP} establish \emph{different} sufficient conditions.
    In \cite{Cover:80} it was shown that separate source and channel coding is generally suboptimal for transmitting correlated sources over MACs. It then directly follows that separate coding is also
    suboptimal for DM MARCs and MABRCs.
    In \Thmref{thm:jointCond_CP} the CPM technique is used for encoding information to the relay, while in \Thmref{thm:jointCondFlip_CP} SW coding concatenated with independent channel coding is used for encoding information to the relay.
    Coupled with the above observation, this implies that the relay decoding constraints of \Thmref{thm:jointCond_CP} are generally looser compared to the relay decoding constraints of \Thmref{thm:jointCondFlip_CP}.
    Using similar reasoning we conclude that the destination decoding constraints of \Thmref{thm:jointCondFlip_CP} are looser compared to the destination decoding constraints of \Thmref{thm:jointCond_CP}
    (as long as coordination is possible, see Remark \ref{rem:CRBCrem}).
    Considering the distribution chains in \eqref{eq:JntJointDist_CP} and \eqref{eq:JntFlipJointDist_CP} we conclude that these two theorems represent different sets of sufficient conditions, and neither are
    special cases of each other nor include one another.
\end{remark}

\begin{remark}
    \Thmref{thm:jointCond_CP} coincides with \Thmref{thm:separationCond} for $\kappa =1$ and no common information:
   Consider the case in which the source pair $(S_1,S_2)$ has no common part, that is $\mT = \phi$, and let $\mQ = \phi$ as well. For an input distribution
    \begin{align*}
        &p(s_1,s_2,w_3,w,v_1,v_2,x_1,x_2,x_3) \nonumber \\
				& \mspace{10mu} = p(s_1,s_2,w_3,w)p(v_1)p(x_1|v_1)p(v_2)p(x_2|v_2)p(x_3|v_1,v_2),
    \end{align*}

\noindent conditions \eqref{bnd:Joint_CP} specialize to conditions \eqref{bnd:sepBased}, and the transmission scheme of \Thmref{thm:jointCond_CP} (see Appendix \ref{annex:jointProof_CP})
specializes to a separation-based achievability scheme of \Thmref{thm:separationCond} for $\kappa = 1$, under these assumptions.
\end{remark}

\begin{remark} \label{rem:JointReduce}
    In both \Thmref{thm:jointCond_CP} and \Thmref{thm:jointCondFlip_CP} the conditions stemming from the CPM technique can be specialized to the sufficient conditions of \cite[Thm. 1]{Cover:80} derived for a MAC.
    In \Thmref{thm:jointCond_CP}, letting $\mV_1= \mV_2= \mX_3= \mW_3=\phi$, specializes the relay conditions in \eqref{bnd:Joint_rly_S1_CP}--\eqref{bnd:Joint_rly_S1S2_2_CP} to the ones in \cite[Thm. 1]{Cover:80}
    with $Y_3$ as the destination.
    In \Thmref{thm:jointCondFlip_CP}, letting $\mX_3=\mW=\phi$, specializes the destination conditions in \eqref{bnd:JointFlip_dst_S1_CP}--\eqref{bnd:JointFlip_dst_S1S2_2_CP} to the ones in \cite[Thm. 1]{Cover:80}
    with $Y$ as the destination.
\end{remark}

\begin{remark} \label{rem:CRBCrem}
\Thmref{thm:jointCond_CP} is optimal in some scenarios: consider the cooperative relay-broadcast channel (CRBC) depicted in
Figure \ref{fig:CBRCsideInfo}. This model is a special case of a MARC obtained when there is a single source terminal. For the CRBC with correlated relay and destination side information, we can identify exactly the optimal
source-channel rate using \Thmref{thm:separationCond} and \Thmref{thm:OuterGeneralMABRC}. This result was previously obtained in \cite{ErkipGunduz:07}:

\begin{figure}[ht]
    \centering
    \scalebox{0.43}{\includegraphics{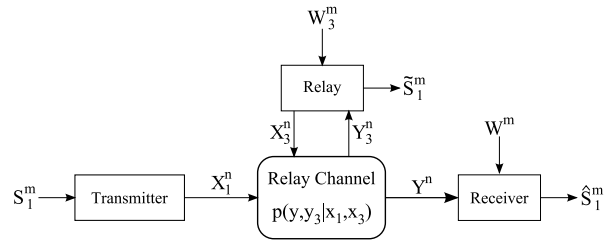}}
    \caption{The cooperative relay-broadcast channel with correlated side information. $\tilde{S}_1^m$ and $\hat{S}_1^m$ are the estimates of $S_1^m$ at the relay and destination, respectively.}
    \label{fig:CBRCsideInfo}
\end{figure}

\begin{corollaryA}  (\cite[Thm. 3.1]{ErkipGunduz:07}) \label{cor:relayJoint}
   For the CRBC with relay and receiver side information, source-channel rate $\kappa$ is achievable if
    \begin{subequations} \label{eq:relayJoint}
        \begin{eqnarray}
            H(S_1|W_3) &<& \kappa I(X_1;Y_3|X_3) \label{eq:RelayJnt_cond} \\
            H(S_1|W) &<& \kappa I(X_1,X_3;Y), \label{eq:DestJnt_cond}
        \end{eqnarray}
    \end{subequations}

\noindent for some input distribution $p(s_1, w_3, w)p(x_1, x_3)$. Conversely, if rate $\kappa$ is achievable then the conditions in \eqref{eq:relayJoint} are satisfied with $<$ replaced by $\leq$
for some input distribution $p(s_1, w_3, w)p(x_1, x_3)$.
\end{corollaryA}

\begin{proof}
    The achievability follows from \Thmref{thm:separationCond} by assigning $X_3=V_1$ and $\mathcal{S}_2 =\mathcal{X}_2 =\mathcal{V}_2=\phi$. The converse follows from \Thmref{thm:OuterGeneralMABRC}.
\end{proof}


For source-channel rate $\kappa=1$, the conditions in \eqref{eq:relayJoint} can also be obtained from \Thmref{thm:jointCond_CP} by letting $V_1=X_3$, $\mathcal{S}_2 =\mathcal{X}_2 =\mathcal{V}_2= \mT = \mQ = \phi$
and considering an input distribution independent of the sources. Observe that \Thmref{thm:jointCondFlip_CP} is not optimal for the CRBC:
consider the conditions in \Thmref{thm:jointCondFlip_CP} with $\mathcal{S}_2=\mathcal{X}_2=\mT = \mQ =\phi$. For this assignment
we obtain the following sufficient conditions:
  \begin{subequations} \label{eq:relayFlipJoint}
    \begin{eqnarray}
        H(S_1|W_3) &<& I(X_1;Y_3|X_3,S_1) \label{eq:RelayFlipJnt_cond} \\
        H(S_1|W) &<& I(X_1,X_3;Y|W), \label{eq:DestFlipJnt_cond}
    \end{eqnarray}

\noindent for some input distribution that factorizes as
\begin{equation}
\label{eqn:input_fac}
    p(s_1, w_3, w)p(x_1|s_1)p(x_3|s_1).
\end{equation}
 \end{subequations}

Note that the RHSs of  \eqref{eq:RelayFlipJnt_cond} and \eqref{eq:DestFlipJnt_cond} are smaller than or equal to the RHSs in \eqref{eq:RelayJnt_cond} and \eqref{eq:DestJnt_cond}, respectively. Moreover,
not all joint input distributions that are feasible for \cite[Thm. 3.1]{ErkipGunduz:07} are also feasible with \eqref{eqn:input_fac}. Hence, the conditions obtained from \Thmref{thm:jointCondFlip_CP}
for the CRBC setup with $\kappa = 1$ are stricter than those obtained from \Thmref{thm:jointCond_CP}, further
illustrating the fact that the two sets of sufficient conditions are not equivalent.
We conclude that the downside of using CPM to the destination as applied in this work is that it places constraints on the distribution chain, thereby constraining the achievable coordination between the sources and the relay.
Due to this restriction, when there is only a single source, the joint distributions of the source and the relay ($X_1$ and $X_3$) permitted for the scheme of \Thmref{thm:jointCondFlip_CP}
do not exhaust the entire space of joint distributions, and as a result, the source-channel sufficient conditions obtained from  \Thmref{thm:jointCondFlip_CP} are generally more restrictive than those obtained
from \Thmref{thm:jointCond_CP} for the single source case.
However, in the case of a MARC it is not possible to determine whether either of the schemes is universally better than the other.
\end{remark}

\begin{remark}
    Note that in \Thmref{thm:jointCond_CP}, it is not useful to generate the relay channel input statistically dependent on the common information, that is, on the auxiliary RV $Q$. To understand why, recall that in \Thmref{thm:jointCond_CP} SW source coding is used for encoding information to the destination, while CPM is used for encoding
information to the relay. The optimality of SW encoding \cite{SW:73} implies that letting the decoder know the common information  will not improve
 the constraints for the source decoder at the destination, as these are based on the SW decoder (see \eqref{eq:DestJntSourceCondE678_CP}).
 Moreover, note that even though CPM is used for encoding information to the relay, sending common
 information via the relay channel input will not improve the decoding constraints at the relay. This follows from the fact that in the DF
 scheme cooperation information is used with a delay of one block. Therefore, common information at the relay channel input corresponds to the source
 sequences of the {\em previous} block, which cannot improve the decoding of the source sequences of the {\em current} block at the relay, in contrast
 to \Thmref{thm:jointCondFlip_CP}. We conclude that in \Thmref{thm:jointCond_CP} we cannot benefit from generating the relay channel input statistically dependent on the common information.
%
%
%
\end{remark}

\begin{remark}
    In both \Thmref{thm:jointCond_CP} and \Thmref{thm:jointCondFlip_CP} we used a combination of SW coding and CPM. Since CPM can generally support higher source-channel rates, a natural question that arises is whether it is possible
    to design a scheme based only on CPM, namely to encode both the cooperation information forwarded by the relay (together with the sources), and the new information transmitted from the sources, using a superposition CPM scheme. This approach cannot be implemented in the framework of the
    current paper. This follows as the current work uses decoding  based on joint typicality, but joint typicality does not apply to different blocks of the same RV. For example, we cannot test the
    joint typicality of $\svec_b$ and $\svec_{b+1}$, as they belong to different time blocks. Using a CPM-only scheme would require us to carry out such tests. For example, consider the case in which the
    source pair $(S_1,S_2)$ has no common part, that is $\mT = \phi$, and also let $\mQ = \phi$. Using the CPM technique for sending information to both the relay and the destination would lead to
    the following relay decoding rule: assume that the relay knows $(\svec_{1,b-1},\svec_{2,b-1})$ at the end of block $b-1$. The relay decodes $(\svec_{1,b}, \svec_{2,b})$,
    by looking for a unique pair $(\tilde{\svec}_{1}, \tilde{\svec}_{2})\in\mS_1^n \times \mS_2^n$ such that:

\begin{align}
    & (\tilde{\svec}_{1}, \tilde{\svec}_{2}, \xvec_1(\tilde{\svec}_{1}, \svec_{1,b-1}), \xvec_2(\tilde{\svec}_{2}, \svec_{2,b-1}), \svec_{1,b-1}, \nonumber \\
		& \qquad \svec_{2,b-1}, \xvec_3(\svec_{1,b-1}, \svec_{2,b-1}), \wvec_{3,b}, \yvec_{3,b})
    \in \styp.
\label{eq:RelayJntDecTypeCPMBoth}
\end{align}

\noindent Note that $(\tilde{\svec}_{1}, \tilde{\svec}_{2})$ and $(\svec_{1,b-1}, \svec_{2,b-1})$ can not be jointly typical since they correspond to different
block indices: $(\tilde{\svec}_{1}, \tilde{\svec}_{2})$ corresponds to block $b$, while $(\svec_{1,b-1}, \svec_{2,b-1})$ corresponds to block $b-1$, and hence, they are independent of each other.
Similarly, the destination would require to check typicality across different blocks.
%

 We conclude that a CPM-only scheme cannot be used together with a joint typicality decoder.
 It may be possible to construct schemes based on a different decoder, or to implement CPM through intermediate RVs to overcome this difficulty, but these topics are left for future research.
\end{remark}

\begin{remark}
    A comparison of the decoding rules of \Thmref{thm:jointCond_CP} (see Appendix \ref{annex:jointProof_CP_Decoding}) and \Thmref{thm:jointCondFlip_CP}
    (see Appendix \ref{annex:jointProofFlip_CP_Decoding}) reveals a difference in the side information block indices used to assist in decoding at the relay and the destination. The decoding rules of \Thmref{thm:jointCond_CP}
    use side information block with the same index as that of the received vector, while the decoding rules of \Thmref{thm:jointCondFlip_CP} use side information block with an index earlier than that of
    the received vector. The difference stems from the fact that in \Thmref{thm:jointCond_CP}  cooperation between the relay and the sources is achieved via auxiliary RVs which represent bin indices, while
    in \Thmref{thm:jointCondFlip_CP} the cooperation is based on the source sequences.
In the DF scheme cooperation information is used with a delay of one block. Therefore, when cooperation is based on the source sequences (\Thmref{thm:jointCondFlip_CP}), then
the side information from the previous block is used for decoding since this is the side information that is correlated with the source sequences.
\end{remark}



\section{Conclusions} \label{sec:conclusions}

In this paper we considered transmission of arbitrarily correlated sources over MARCs and MABRCs with correlated side information at both the relay and the destination.
We first derived an  achievable source-channel rate for MARCs based on operational separation, which applies directly to MABRCs as well.
This result is established by using an irregular encoding scheme for the channel code. We  also showed that for both MABRCs and MARCs regular encoding is more restrictive than irregular encoding.
additionally, we obtained necessary conditions for the achievability of source-channel rates.

Then, we considered phase fading and Rayleigh fading MARCs with side information and identified conditions under which informational separation is optimal for these channels.
Conditions for the optimality of informational separation for fading MABRCs were also obtained.
The importance of this result lies in the fact that it supports
a modular system design (separate design of the source and channel codes) while achieving the optimal end-to-end performance. We note here that this is the
first time that optimality of separation is shown for a MARC or a MABRC configuration.

Lastly, we considered joint source-channel coding for DM MARCs and MABRCs for source-channel rate $\kappa=1$.
We presented two new joint source-channel coding schemes for which use a combination of SW source coding and joint source-channel coding based on CPM.
While in the first scheme CPM is used for encoding information to the relay and SW coding is used for encoding information to the destination; in the second scheme SW coding is used for
encoding information to the relay and CPM is used for encoding information to the destination. The different combinations of SW coding and CPM enable flexibility in the system design by
choosing one of the two schemes according to the qualities of the side information sequences and received signals at the relay and the destination.
In particular, the first scheme generally has looser decoding constraints at the relay, and therefore it is better when the source-relay link is the bottleneck,
while the second scheme generally has looser decoding constraints at the destination, and is more suitable to scenarios where the source-destination link is more limiting.

\appendices
\numberwithin{equation}{section}

\vspace{-0.1cm}
\section{Proof of Theorem \thmref{thm:separationCond}} \label{annex:proofSeparation}
Fix a distribution $p(v_1)p(x_1|v_1)p(v_2)p(x_2|v_2)p(x_3|v_1,v_2)$.
\subsection{Codebook construction}
For $i=1,2$, assign every $\svec_i \in \mS_i^m$ to one of $2^{mR_i^r}$ bins
independently according to a uniform distribution on $\msgCal_i^r \triangleq \{1,2,\dots,2^{mR_i^r}\}$. We refer to these two sets as the {\em relay bins}. Denote these assignments by $f_i^r$.
    Independent from the relay bin assignments, for $i=1,2$, assign every $\svec_i \in \mS_i^m$ to one of $2^{mR_i^d}$ bins independently according to a uniform distribution on $\msgCal_i^d \triangleq \{1,2,\dots,2^{mR_i^d}\}$.
We refer to these two sets as the {\em destination bins}. Denote these assignments by $f_i^d$.

    Next, generate a superposition channel codebook  with blocklength $n$, rates $\chR_i^d = \frac{1}{\kappa}R_i^d$, $i=1,2$, auxiliary vectors
    $\vvec_i(\msg_i^d), \msg_i^d \in \msgCal_i^d $, $i=1,2$, and channel codewords $\xvec_i(\msg_i^r, \msg_i^d)$, $(\msg_i^r,\msg_i^d) \in \msgCal_i^r \times \msgCal_i^d$, $i=1,2$, as detailed in
    \cite[Appendix A]{Kramer:2005}.


\subsection{Encoding}
\ifthenelse{\boolean{SquizFlag}}{}{}

Consider the  sequences and side information $s^{Bm}_{i,1} \in \mS^{Bm}_{i}, i=1,2$, $w_{3,1}^{Bm}\in\mW_3^{Bm}$, and $w^{Bm}\in\mW^{Bm}$, all of length $Bm$.
Partition each sequence into $B$ length $m$ subsequences, $\svec_{i,b}$, $i=1,2$, $\wvec_{3,b}$, and $\wvec_b$,  $b=1,2,\dots,B$.
A total of $Bm$ source samples are transmitted in $B+1$ blocks of $n$ channel symbols each.
For any fixed $(m,n)$ with $n \leq \kappa m  $, we can achieve a rate arbitrarily close to $\kappa = n/m$ by increasing $B$, i.e,
$(B+1)n/Bm \rightarrow \kappa$ as $B \rightarrow \infty$.

At block $1$, transmitter $i, i=1,2$, observes source subsequence $\svec_{i,1}$ and finds its corresponding relay bin index
$\msg_{i,1}^r  = f_i^r(\svec_{i,1})\in \msgCal_u^r$. It transmits the channel codeword $\xvec_i(\msg_{i,1}^r,1)$.
In block $b, b=2,\dots,B$, source terminal $i$ transmits the channel codeword $ \xvec_i(\msg_{i,b}^r,\msg_{i,b-1}^d)$,
where $\msg_{i,b}^r = f_i^r(\svec_{i,b})\in \msgCal_i^r$, and $\msg_{i,b-1}^d = f_i^d(\svec_{i,b-1}) \in \msgCal_i^d$.
In block $B+1$, the source terminal transmits $\xvec_i(1,\msg_{i,B}^d)$.

At block $b=1$, the relay simply transmits $\xvec_3(1,1)$.
Assume that at block $b, b=2,\dots,B,B+1$, the relay estimates $(\svec_{1,b-1},\svec_{2,b-1})$. Let $(\tsvec_{1,b-1},\tsvec_{2,b-1})$ denote the estimates.
The relay then finds the corresponding destination bin indices $\tilde{\msg}_{i,b-1}^d \in \msgCal_i^d, i=1,2$, and transmits the channel codeword
$\xvec_3(\tilde{\msg}_{1,b-1}^d,\tilde{\msg}_{2,b-1}^d)$.

\ifthenelse{\boolean{SquizFlag}}{}{}

\subsection{Decoding}

The relay decodes the source sequences sequentially trying to
reconstruct source blocks $\svec_{i,b}, i=1,2$, at the end of channel block $b$ as follows:
Let $(\tsvec_{1,b-1},\tsvec_{2,b-1})$ be the estimates of $(\svec_{1,b-1},\svec_{2,b-1})$ obtained
at the end of block $b-1$. Applying $f_1^d$ and $f_2^d$, the relay finds the corresponding destination bin indices $(\tmsg_{1,b-1}^d, \tmsg_{2,b-1}^d)$.
At time $b$ the relay channel decoder decodes $(\msg_{1,b}^r, \msg_{2,b}^r)$
by looking for a unique pair $(\tmsg_{1}^r, \tmsg_{2}^r) \in \mU_1^r\times\mU_2^r$ such that:
\begin{align}
    & \mspace{-12mu} \big(\vvec_1 \mspace{-1mu} (\tmsg_{1,b-1}^d), \vvec_2 \mspace{-1mu}(\tmsg_{2,b-1}^d), \xvec_1\mspace{-1mu}(\tmsg_{1}^r, \tmsg_{1,b-1}^d), \xvec_2\mspace{-1mu}(\tmsg_{2}^r, \tmsg_{2,b-1}^d), \nonumber \\
		& \mspace{-8mu} \xvec_3 \mspace{-2mu}(\tmsg_{1,b-1}^d, \mspace{-1mu}\tmsg_{2,b-1}^d \mspace{-1mu}), \mspace{-1mu} \yvec_{3,b} \mspace{-1mu} \big) \mspace{-4mu} \in \mspace{-4mu} \styp \mspace{-2mu}(\mspace{-1mu}V_1, \mspace{-2mu}V_2, \mspace{-1mu}X_1, \mspace{-1mu}X_2, \mspace{-1mu}X_3, \mspace{-2mu}Y_3 \mspace{-1mu}).
    \label{eq:RelayChanDecType}
\end{align}

The decoded relay bin indices, denoted $(\tmsg_{1,b}^r, \tmsg_{2,b}^r)$, are then given to the relay source decoder,
which estimates
$(\svec_{1,b}, \svec_{2,b})$.
The relay source decoder declares $(\tsvec_{1},\tsvec_{2}) \in \mS_1^m \times \mS_2^m$ as the decoded sequences if it is the
unique pair of sequences that satisfies $f_1^r(\tsvec_{1})= \tmsg_{1,b}^r, f_2^r(\tsvec_{2})= \tmsg^r_{2,b}$ and $(\tsvec_{1},\tsvec_{2},\wvec_{3,b}) \in \stypm(S_1,S_2,W_3)$. The decoded sequences are denoted by
$(\tsvec_{1,b},\tsvec_{2,b})$.

Decoding at the destination is done using backward decoding. The destination node waits until the end of channel
block $B+1$. It first tries to decode $(\svec_{1,B},\svec_{2,B})$ using the received signal at channel block $B+1$ and its
side information $\wvec_{B}$. Going backwards from the last channel block to the first, we assume that the destination has estimates
$(\hsvec_{1,b+1},\hsvec_{2,b+1})$ of $(\svec_{1,b+1},\svec_{2,b+1})$, and consider decoding of $(\svec_{1,b},\svec_{2,b})$. From $\hsvec_{i,b+1}, i=1,2$, the destination finds
the relay bin indices $\hmsg_{i,b+1}^r = f_i^r(\hsvec_{i,b+1})$. At block $b+1$ the destination channel decoder first estimates the destination
bin indices $(\msg_{1,b}^d, \msg_{2,b}^d)$
by looking for
a unique pair $(\hmsg_{1}^d, \hmsg_{2}^d)\in\mU_1^d\times\mU_2^d$ such that:
\begin{align}
    & \mspace{-10mu} (\vvec_1(\hmsg_{1}^d), \vvec_2(\hmsg_{2}^d), \xvec_1(\hmsg_{1,b+1}^r, \hmsg_{1}^d), \xvec_2(\hmsg_{2,b+1}^r, \hmsg_{2}^d), \nonumber \\
		& \quad \xvec_3(\hmsg_{1}^d, \hmsg_{2}^d), \yvec_{b+1}) \in \styp(V_1,V_2,X_1,X_2,X_3,Y_3).
\label{eq:DestChanDecType}
\end{align}

The decoded destination bin indices, denoted $(\hmsg_{1,b}^d, \hmsg_{2,b}^d)$, are then given to the destination source decoder,
which estimates the source sequences $(\svec_{1,b}, \svec_{2,b})$. 
The destination source decoder declares $(\hsvec_{1},\hsvec_{2}) \in \mS_1^m \times \mS_2^m$ as the decoded sequences if it is the unique pair of sequences that satisfies
$f_1^d(\hsvec_{1})= \hmsg_{1,b}^d, f_2^d(\hsvec_{2})= \hmsg_{2,b}^d$ and $(\hsvec_{1},\hsvec_{2},\wvec_{b}) \in \stypm(S_1,S_2,W)$. The decoded sequences are denoted by $(\hsvec_{1,b},\hsvec_{2,b})$.

\ifthenelse{\boolean{SquizFlag}}{}{}

\ifthenelse{\boolean{SquizFlag}}{}{}

\ifthenelse{\boolean{SquizFlag}}{}{}


\section{Proof of Theorem \thmref{thm:jointCond_CP}} \label{annex:jointProof_CP}

\begin{figure*}[!b]
\normalsize
\setcounter{MYtempeqncnt}{\value{equation}}
\setcounter{equation}{2}

\vspace*{4pt}
\hrulefill
\vspace{-0.3cm}

    \begin{align}
    \bar{P}_{r}^{(n)}
    & \triangleq \mspace{-15mu} \sum_{(\msg_{1,b-1}, \msg_{2,b-1}) \in \msgCal_1 \times \msgCal_2}{\mspace{-50mu} p(\msg_{1,b-1}, \msg_{2,b-1})}
        \sum_{(\svec_{1,b},\svec_{2,b}) \in \mS_1^n \times \mS_2^n}{\mspace{-40mu} p(\svec_{1,b},\svec_{2,b})}
            \Pr \big\{ E_b^{r}(\svec_{1,b},\svec_{2,b};\msg_{1,b-1}, \msg_{2,b-1}) \big\}. \label{eq:RelayDecErrProbDef_Basic_CP}
\end{align}

\vspace*{4pt}
\hrulefill
\vspace{-0.3cm}

\begin{subequations} \label{eq:RelayDecErrProbDef_CP}
\begin{align}
    & \sum_{(\svec_{1,b},\svec_{2,b}) \in \mS_1^n \times \mS_2^n}{\mspace{-24mu} p(\svec_{1,b},\svec_{2,b})} \Pr \big\{ E_b^{r}(\svec_{1,b},\svec_{2,b};\msg_{1,b-1}, \msg_{2,b-1}) \big\} \nonumber \\
    & \mspace{50mu} \leq \mspace{-24mu} \sum_{(\svec_{1,b},\svec_{2,b}, \wvec_{3,b}) \notin \styp(S_1,S_2,W_3)}{\mspace{-54mu} p(\svec_{1,b},\svec_{2,b}, \wvec_{3,b})} \nonumber \\
    & \qquad\qquad   + \!\!\!\! \!\!\!\! \!\!\!\! \!\!\!\! \!\!\!\! \sum_{(\svec_{1,b},\svec_{2,b}, \wvec_{3,b}) \in \styp(S_1,S_2,W_3)}{\mspace{-54mu} \!\!\!\!\!\!\!\!\!\!p(\svec_{1,b},\svec_{2,b},\wvec_{3,b})}
        \Pr \big\{ E_b^{r}(\svec_{1,b},\svec_{2,b};\msg_{1,b-1}, \msg_{2,b-1}) | (\svec_{1,b},\svec_{2,b}, \wvec_{3,b}) \in \styp(S_1,S_2,W_3) \big\} \label{eq:RelayDecErrProbDef_ub_2_CP} \\
%
    & \mspace{50mu} \leq \epsilon + \mspace{-24mu} \sum_{(\svec_{1,b},\svec_{2,b}, \wvec_{3,b}, \tvec_b) \in \styp(S_1,S_2,W_3,T)}{\mspace{-70mu} p(\svec_{1,b},\svec_{2,b}, \wvec_{3,b})}
        \Pr \big\{ E_b^{r}(\svec_{1,b},\svec_{2,b};\msg_{1,b-1}, \msg_{2,b-1}) | \msD_b \big\} \label{eq:RelayDecErrProbDef_AddCommonPart_2_CP},
\end{align}
\end{subequations}

\setcounter{equation}{\value{MYtempeqncnt}}

\end{figure*}

Fix a distribution $p(s_1,s_2,w_3,w)p(q)p(v_1)p(x_1|s_1,v_1,q)$ $p(v_2)p(x_2|s_2,v_2,q)p(x_3|v_1,v_2)p(y_3,y|x_1,x_2,x_3)$.
\subsection{Codebook construction}
For $i=1,2$, assign every $\svec_i \in \mS_i^n$ to one of $2^{nR_i}$ bins independently according to a uniform distribution on $\msgCal_i \triangleq \{1,2,\dots,2^{nR_i}\}$. Denote this assignment by $f_i, i=1,2$.

For the channel codebook, for each $i=1,2$, generate $2^{nR_i}$ codewords $\vvec_i(\msg_i), \msg_i \in \msgCal_i $, by choosing the letters $v_{i,k}(\msg_i), k = 1,2,\dots,n$,
independently according to the distribution $p_{V_i}(v_{i,k}(\msg_i))$.
For each $\tvec \in \mT^n$ generate one length $n$ codeword $\qvec(\tvec)$ by choosing the letters $q_k$ independently with distribution $p_{Q}(q_{k})$,  $k = 1,2,\dots, n$.
For each pair $(\svec_i, \msg_i) \in \mS_i^n \times \msgCal_i, i=1,2$, find the corresponding $\tvec = h_i(\svec_i)$, and generate one length $n$ codeword
$\xvec_i(\svec_i, \msg_i, \qvec(\tvec))$, $\qvec \in \mQ^n$, by choosing the letters $x_{i,k}(\svec_i, \msg_i, \qvec(\tvec))$ independently with distribution
$p_{X_i|S_i,V_i, Q}(x_{i,k}|s_{i,k},v_{i,k}(\msg_i), q_k(\tvec))$, $k = 1,2,\dots, n$.
Finally, generate one length-$n$ relay codeword $\xvec_3(\msg_1,\msg_2)$ for ~each ~pair ~$(\msg_1, \msg_2) \in \msgCal_1 \times \msgCal_2$ ~by ~choosing ~$x_{3,k}(\msg_1,\msg_2)$ ~independently ~with ~distribution
$p_{X_3|V_1,V_2}(x_{3,k}|v_{1,k}(\msg_1),v_{2,k}(\msg_2))$, $k = 1,2,\dots, n$.

\subsection{Encoding} \label{annex:jointProof_CP_Encoding}
    Consider the sequences $s^{Bn}_{i,1} \in \mS^{Bn}_i, i=1,2$, $w_{3,1}^{Bn}\in\mW_3^{Bn}$, and $w^{Bn}\in\mW^{Bn}$, all of length $Bn$.
    Partition each sequence into $B$ length-$n$ subsequences, $\svec_{i,b}$, $i=1,2$, $\wvec_{3,b}$, and $\wvec_b$,  $b=1,2,\dots,B$.
A total of $Bn$ source samples are transmitted in $B+1$ blocks of $n$ channel symbols each.
At block $1$, source terminal $i$, $i=1,2$, finds $\tvec_i = h_i(\svec_{i,1})$, and transmits the channel codeword $\xvec_i(\svec_{i,1}, 1, \qvec(h_i(\svec_{i,1})))$.
At block $b, b=2,\dots,B$, source terminal $i$, $i=1,2$, transmits the channel codeword $\xvec_i(\svec_{i,b}, \msg_{i,b-1}, \qvec(h_i(\svec_{i,b})))$, where $\msg_{i,b-1} = f_i(\svec_{i,b-1}) \in \msgCal_i$
is the bin index of source vector $\svec_{i,b-1}$.
Let $(\avec_1, \avec_2) \in \mS_1^n \times \mS_2^n$ be two sequences generated i.i.d according to $p(\avec_1,\avec_2) = \prod_{k=1}^{n}{p_{S_1,S_2}(a_{1,k}, a_{2,k})}$. These sequences are known to all nodes.
At block $B+1$, source terminal $i$ transmits $\xvec_i(\avec_i,\msg_{i,B}, \qvec(h_i(\avec_i)))$.

At block $b=1$, the relay transmits $\xvec_3(1,1)$.
Assume that at block $b, b=2,\dots,B,B+1$, the relay has estimates $(\tilde{\svec}_{1,b-1}, \tilde{\svec}_{2,b-1})$ of $(\svec_{1,b-1}, \svec_{2,b-1})$. It then finds the corresponding bin indices
$\tilde{\msg}_{i,b-1} = f_i(\tilde{\svec}_{1,b-1}) \in \msgCal_i, i=1,2$, and transmits the channel codeword $\xvec_3(\tilde{\msg}_{1,b-1},\tilde{\msg}_{2,b-1})$.

\subsection{Decoding} \label{annex:jointProof_CP_Decoding}
The relay decodes the source sequences sequentially, trying to reconstruct $(\svec_{1,b}, \svec_{2,b})$ at the end of channel block $b$ as follows:
Let $(\tsvec_{1,b-1},\tsvec_{2,b-1})$ be the estimates of $(\svec_{1,b-1},\svec_{2,b-1})$ obtained at the end of block $b-1$. The relay thus knows the corresponding bin indices $(\tmsg_{1,b-1},\tmsg_{2,b-1})$.
Using this information, its received signal $\yvec_{3,b}$, and the side information $\wvec_{3,b}$, the relay decodes $(\svec_{1,b}, \svec_{2,b})$, by looking for a unique pair $(\tsvec_{1}, \tsvec_{2})\in \mS_1^n\times\mS_2^n$
such that:
\begin{align}
    & \mspace{-8mu} \big(\tsvec_{1}, \tsvec_{2}, \ttvec, \qvec(\ttvec), \vvec_1(\tmsg_{1,b-1}), \vvec_2(\tmsg_{2,b-1}), \xvec_1(\tsvec_{1}, \tmsg_{1,b-1}, \qvec(\ttvec)), \nonumber \\
		& \mspace{10mu}  \xvec_2(\tsvec_{2}, \tmsg_{2,b-1}, \qvec(\ttvec)), \xvec_3(\tmsg_{1,b-1}, \tmsg_{2,b-1}), \wvec_{3,b}, \yvec_{3,b}\big) \nonumber\\
    & \mspace{16mu}  \in \styp (S_1,S_2,T,Q,V_1,V_2,X_1,X_2,X_3,W_3,Y_3),
\label{eq:RelayJntDecType_CP}
\end{align}

\noindent where $\ttvec = h_1(\tsvec_{1}) = h_2(\tsvec_{2})$. Denote the decoded sequence $(\tsvec_{1,b}, \tsvec_{2,b})$.

Decoding at the destination is done using backward decoding. Let $\boldsymbol{\alpha} \in \mW^n$ be a sequence generated i.i.d according to $p_{W|S_1,S_2}(\alpha_k | a_{1,k}, a_{2,k}), k=1,2,\dots,n$.
The destination node waits until the end of channel block $B+1$. It first tries to decode $(\svec_{1,B},\svec_{2,B})$ using the received signal at channel block $B+1$ and $\boldsymbol{\alpha}$.
Going backwards from the last channel block to the first, we assume that the destination has estimates $(\hsvec_{1,b+1},\hsvec_{2,b+1})$ of $(\svec_{1,b+1},\svec_{2,b+1})$, and therefore has the
estimates $\htvec_{b+1} = h_1(\hsvec_{1,b+1}) = h_2(\hsvec_{2,b+1})$ and $\qvec(\htvec_{b+1})$.
At block $b+1$ the destination channel decoder first estimates the destination bin indices $\hmsg_{i,b}, i=1,2$, corresponding to $\svec_{i,b}$, based on its received signal $\yvec_{b+1}$ and the side information $\wvec_{b+1}$,
by looking for a unique pair $(\hmsg_{1}, \hmsg_{2})\in\mU_1\times\mU_2$ such that:
\begin{align}
    & \mspace{-10mu} \big(\hsvec_{1,b+1}, \hsvec_{2,b+1}, \htvec_{b+1}, \qvec(\htvec_{b+1}), \vvec_1(\hmsg_{1}), \vvec_2(\hmsg_{2}), \nonumber \\
		& \xvec_1(\hsvec_{1,b+1}, \hmsg_{1}, \qvec(\htvec_{b+1})), \xvec_2(\hsvec_{2,b+1}, \hmsg_{2}, \qvec(\htvec_{b+1})), \nonumber\\
    & \mspace{15mu} \xvec_3(\hmsg_{1}, \hmsg_{2}), \wvec_{b+1}, \yvec_{b+1}\big) \nonumber \\
		& \mspace{30mu} \in \styp (S_1,S_2,T,Q,V_1,V_2,X_1,X_2,X_3,W,Y).
\label{eq:DestJntChanDecType_CP}
\end{align}

    The decoded destination bin indices, denoted by $(\hmsg_{1,b},\hmsg_{2,b})$, are then given to the destination source decoder,
   which estimates $(\svec_{1,b}, \svec_{2,b})$ by looking for a
    unique pair of sequences $(\hsvec_{1},\hsvec_{2})\in\mS_1^n\times\mS_2^n$ that satisfies $f_1(\hsvec_{1})= \hmsg_{1,b}, f_2(\hsvec_{2})= \hmsg_{2,b}$ and $(\hsvec_{1},\hsvec_{2},\wvec_{b}) \in \styp(S_1,S_2,W)$.
    The decoded sequences are denoted by $(\hsvec_{1,b},\hsvec_{1,b})$.


\subsection{Error Probability Analysis} \label{annex:jointErrAnalysis_CP}


We start with the relay error probability analysis.
Let $E_b^{r}(\svec_{1,b},\svec_{2,b};\msg_{1,b-1}, \msg_{2,b-1})$ denote the relay decoding error event in block $b$, assuming $(\msg_{1,b-1}, \msg_{2,b-1})$ are available at the relay,
and $(\svec_{1,b},\svec_{2,b})$ are the source sequences at block $b$.
Thus, this error event is the event that $(\tsvec_{1,b},\tsvec_{2,b}) \ne (\svec_{1,b},\svec_{2,b})$. Let $\msD_b$ denote the event that $(\svec_{1,b},\svec_{2,b}, \wvec_{3,b}, \tvec_b) \in \styp(S_1,S_2, W_3,T)$.
The average decoding error probability at the relay in block $b$, $\bar{P}_{r}^{(n)}$, is defined in \eqref{eq:RelayDecErrProbDef_Basic_CP} at the bottom of the page.
\noindent In the following we show that the inner sum in \eqref{eq:RelayDecErrProbDef_Basic_CP} can be upper bounded independently of $(\msg_{1,b-1}, \msg_{2,b-1})$.
Therefore, for any {\em fixed} value of $(\msg_{1,b-1}, \msg_{2,b-1})$ we have \eqref{eq:RelayDecErrProbDef_CP} at the bottom of the page,
\noindent where \eqref{eq:RelayDecErrProbDef_ub_2_CP} follows from the union bound and
\eqref{eq:RelayDecErrProbDef_AddCommonPart_2_CP}
follows from the AEP \cite[Ch. 5.1]{YeungBook}, for sufficiently large $n$,  and
as $\tvec_b$ is a deterministic function of $(\svec_{1,b},\svec_{2,b})$. This
deterministic relationship implies that  $(\svec_{1,b},\svec_{2,b},\wvec_{3,b})\in\styp(S_1,S_2,W_3)$ if and only if $(\svec_{1,b},\svec_{2,b},\wvec_{3,b},\tvec_b)\in\styp(S_1,S_2,W_3,T)$.
Note also that  \eqref{eq:RelayDecErrProbDef_AddCommonPart_2_CP} follows similarly to \cite[Eq. (16)]{Cover:80}.
Next, we show that for $(\svec_{1,b},\svec_{2,b}, \wvec_{3,b}, \tvec_b) \in \styp(S_1,S_2, W_3,T)$, the summands in \eqref{eq:RelayDecErrProbDef_AddCommonPart_2_CP} can be upper bounded independently
of $(\svec_{1,b},\svec_{2,b}, \wvec_{3,b})$, for any fixed value of $(\msg_{1,b-1}, \msg_{2,b-1})$.

\setcounter{equation}{4}

Let $\eps_0$, $\eps_1$ and $\eps_2$ be positive numbers such that $\eps_0 \ge \eps_2 \ge \eps_1 > \eps$ and $\eps_0 \rightarrow 0$ as $\eps \rightarrow 0$.
Assuming correct decoding at block $b-1$ (hence $(\msg_{1,b-1}, \msg_{2,b-1})$ are available at the relay), we define the following events:
\begin{align*}
    E_1^{r} \triangleq & \big\{(\svec_{1,b}, \svec_{2,b}, \tvec_b, \Qvec(\tvec_b), \Vvec_1(\msg_{1,b-1}), \Vvec_2(\msg_{2,b-1}), \\
		& \quad \Xvec_1(\svec_{1,b}, \msg_{1,b-1}, \Qvec(\tvec_b)), \Xvec_2(\svec_{2,b}, \msg_{2,b-1}, \Qvec(\tvec_b)),  \\
    & \qquad \Xvec_3(\msg_{1,b-1}, \msg_{2,b-1}), \wvec_{3,b}, \Yvec_{3,b}) \notin \styp \big\}, \\
    E_2^r  \triangleq & \big\{\exists (\tilde{\svec}_{1}, \tilde{\svec}_{2}) \in \mS_1^n\times\mS_2^n: \\
			& \quad (\tilde{\svec}_{1}, \tilde{\svec}_{2}) \neq (\svec_{1,b}, \svec_{2,b}), \ttvec = h_1(\tsvec_1) = h_2(\tsvec_2), \\
			& \qquad \big(\tilde{\svec}_{1}, \tilde{\svec}_{2}, \ttvec, \Qvec(\ttvec), \Vvec_1(\msg_{1,b-1}), \Vvec_2(\msg_{2,b-1}), \\
    & \qquad\quad \Xvec_1(\tilde{\svec}_{1}, \msg_{1,b-1},\Qvec(\ttvec)), \Xvec_2(\tilde{\svec}_{2}, \msg_{2,b-1},\Qvec(\ttvec)), \\
		& \qquad \qquad \Xvec_3(\msg_{1,b-1}, \msg_{2,b-1}), \wvec_{3,b}, \Yvec_{3,b} \big) \in \styp \big\}.
\end{align*}

\noindent From the AEP \cite[Ch. 5.1]{YeungBook}, for sufficiently large $n$, $\Pr \LL{\{} E^{r}_1 |\msD_b \RR{\}} \le \eps$. Therefore we can bound
\begin{align}
    & \Pr \big\{ E_b^{r}(\svec_{1,b},\svec_{2,b};\msg_{1,b-1}, \msg_{2,b-1}) | \msD_b \big\} \nonumber \\
		& \qquad \qquad \qquad \qquad \quad \leq \epsilon +  \Pr \big\{ E_2^{r} | (E_1^{r})^{c} \big\}.
\label{eq:RelayDecErrProb1_CP}
\end{align}

The event $E_2^r$ is the union of the following events:
\begin{align*}
    E_{21}^{r} \triangleq & \big\{\exists \tilde{\svec}_{1}\in\mS_1^n: \tilde{\svec}_{1} \neq \svec_{1,b}, h_1(\tsvec_1) = h_2(\svec_{2,b}) = \tvec_b, \\
		& \quad \big( \tilde{\svec}_{1}, \svec_{2,b}, \tvec_b, \Qvec(\tvec_b), \Vvec_1(\msg_{1,b-1}), \Vvec_2(\msg_{2,b-1}),\\
    & \qquad \Xvec_1(\tilde{\svec}_{1}, \msg_{1,b-1}, \Qvec(\tvec_b)), \Xvec_2(\svec_{2,b}, \msg_{2,b-1}, \Qvec(\tvec_b)), \\
		& \qquad \quad \Xvec_3(\msg_{1,b-1}, \msg_{2,b-1}), \wvec_{3,b}, \Yvec_{3,b} \big) \in \styp \big\}   \\
    E_{22}^r \triangleq & \big\{\exists \tilde{\svec}_{2} \in\mS_2^n: \tilde{\svec}_{2} \neq \svec_{2,b}, h_1(\svec_{1,b}) = h_2(\tsvec_2) = \tvec_b, \\
		& \quad \big( \svec_{1,b}, \tilde{\svec}_{2}, , \tvec_b, \Qvec(\tvec_b),
    \Vvec_1(\msg_{1,b-1}), \Vvec_2(\msg_{2,b-1}),\\
    & \qquad \Xvec_1(\svec_{1,b}, \msg_{1,b-1}, \Qvec(\tvec_b)), \Xvec_2(\tilde{\svec}_{2}, \msg_{2,b-1}, \Qvec(\tvec_b)), \\
		& \qquad \quad \Xvec_3(\msg_{1,b-1}, \msg_{2,b-1}), \wvec_{3,b}, \Yvec_{3,b} \big) \in \styp \big\} \\
    E_{23}^r \triangleq & \big\{\exists (\tilde{\svec}_{1},\tsvec_2)\in\mS_1^n\times\mS_2^n: \\
		& \quad \tilde{\svec}_{1} \neq \svec_{1,b},  \tilde{\svec}_{2} \neq \svec_{2,b}, h_1(\tsvec_1) = h_2(\tsvec_2) = \tvec_b, \\
		& \qquad \big( \tilde{\svec}_{1}, \tilde{\svec}_{2}, \tvec_b, \Qvec(\tvec_b), \Vvec_1(\msg_{1,b-1}), \Vvec_2(\msg_{2,b-1}), \\
		& \qquad \quad \Xvec_1(\tilde{\svec}_{1}, \msg_{1,b-1}, \Qvec(\tvec_b)),  \Xvec_2(\tilde{\svec}_{2}, \msg_{2,b-1}, \Qvec(\tvec_b)), \\
		& \qquad \qquad \Xvec_3(\msg_{1,b-1}, \msg_{2,b-1}), \wvec_{3,b}, \Yvec_{3,b} \big)
    \in \styp \big\} \\
    E_{24}^r \triangleq & \big\{\exists (\tilde{\svec}_{1},\tsvec_2)\in\mS_1^n\times\mS_2^n: \tilde{\svec}_{1} \neq \svec_{1,b},  \tilde{\svec}_{2} \neq \svec_{2,b}, \\
		& \quad h_1(\tsvec_1) = h_2(\tsvec_2) = \ttvec\ne \tvec_b, \Qvec(\ttvec) \ne \Qvec(\tvec_b), \\
		& \qquad  \big( \tilde{\svec}_{1}, \tilde{\svec}_{2}, \ttvec, \Qvec(\ttvec), \Vvec_1(\msg_{1,b-1}),\Vvec_2(\msg_{2,b-1}), \\
		& \qquad \quad \Xvec_1(\tilde{\svec}_{1}, \msg_{1,b-1}, \Qvec(\ttvec)), \Xvec_2(\tilde{\svec}_{2}, \msg_{2,b-1}, \Qvec(\ttvec)), \\
		& \qquad \qquad \Xvec_3(\msg_{1,b-1}, \msg_{2,b-1}), \wvec_{3,b}, \Yvec_{3,b} \big) \in \styp \big\} \\
    E_{25}^r \triangleq & \big\{\exists (\tilde{\svec}_{1},\tsvec_2)\in\mS_1^n\times\mS_2^n: \tilde{\svec}_{1} \neq \svec_{1,b}, \tilde{\svec}_{2} \neq \svec_{2,b}, \\
		& \quad h_1(\tsvec_1) = h_2(\tsvec_2)= \ttvec \ne \tvec_b, \Qvec(\ttvec) = \Qvec(\tvec_b),\\
    & \qquad \big( \tilde{\svec}_{1}, \tilde{\svec}_{2}, \ttvec, \Qvec(\ttvec),\Vvec_1(\msg_{1,b-1}),\Vvec_2(\msg_{2,b-1}), \\
		& \qquad \quad \Xvec_1(\tilde{\svec}_{1}, \msg_{1,b-1}, \Qvec(\ttvec)), \Xvec_2(\tilde{\svec}_{2}, \msg_{2,b-1}, \Qvec(\ttvec)), \\
		& \qquad  \qquad \Xvec_3(\msg_{1,b-1}, \msg_{2,b-1}), \wvec_{3,b}, \Yvec_{3,b} \big) \in \styp \big\}.
\end{align*}

\noindent Hence, by the union bound it follows that
   $ \Pr \big\{ E_2^{r} | (E_1^{r})^{c} \big\} = \sum_{j=1}^{5}{\Pr \big\{ E_{2j}^{r} | (E_1^{r})^{c} \big\}}$.
%
To bound $\Pr \big\{ E_{21}^{r} | (E_1^{r})^{c} \big\}$ we first define the event $E_{21}^{r}(\tilde{\svec}_1)$ as follows
\begin{align}
     & \mspace{-15mu} E_{21}^{r}(\tilde{\svec}_1) \mspace{-2mu} \triangleq \mspace{-2mu} \big\{ h_1(\tsvec_1) = h_2(\svec_{2,b})=\tvec_b, \nonumber \\
		& \mspace{60mu} \big(\tilde{\svec}_{1}, \svec_{2,b}, \tvec_b, \Qvec(\tvec_b), \Vvec_1(\msg_{1,b-1}), \Vvec_2(\msg_{2,b-1}), \nonumber \\
    & \mspace{68mu} \Xvec_1(\tilde{\svec}_{1}, \msg_{1,b-1}, \Qvec(\tvec_b)), \Xvec_2(\svec_{2,b}, \msg_{2,b-1}, \Qvec(\tvec_b)), \nonumber \\
		& \mspace{76mu} \Xvec_3(\msg_{1,b-1}, \msg_{2,b-1}), \wvec_{3,b}, \Yvec_{3,b}\big) \in \styp \big\}.
\label{eq:Relay_E21SpecificDef_CP}
\end{align}

\noindent Recalling that for $E_{21}^r$ then $\tsvec_1 \ne \svec_{1,b}$, we have
\begin{equation}
	\Pr \big\{ E_{21}^{r} | (E_1^{r})^{c} \big\} =
    \mspace{-60mu} \sum_{\mspace{20mu} \substack{\tilde{\svec}_1 \neq \svec_{1,b}, \\ \tilde{\svec}_1 \in \styp(S_1|\svec_{2,b}, \wvec_{3,b}, \tvec_b)}}{\mspace{-45mu} \Pr \big\{ E_{21}^{r}(\tilde{\svec}_1) | (E_1^{r})^{c} \big\}}.
\label{eq:RelayE21_Bound}
\end{equation}

\noindent Note that in \eqref{eq:RelayE21_Bound} we consider $\tilde{\svec}_1 \in \styp(S_1|\svec_{2,b}, \wvec_{3,b}, \tvec_b) \big\}$,
as otherwise $\Pr\big\{E_{21}^{r}(\tilde{\svec}_1)| (E_1^{r})^{c}\big\}=0$. Therefore, in the following we upper bound $\Pr \Big\{ E_{21}^{r}(\tilde{\svec}_1) \Big| (E_1^{r})^{c}, \tilde{\svec}_1 \in \styp(S_1|\svec_{2,b}, \wvec_{3,b}, \tvec_b) \Big\}$ via an expression that~is~independent of $\tsvec_1$.
To reduce clutter, let us denote $\svec_{2,b}$, $\tvec_{b}$, $\qvec(\tvec_b)$, $\vvec_1(\msg_{1,b-1})$, $\vvec_2(\msg_{2,b-1})$, $\xvec_1(\tilde{\svec}_{1}, \msg_{1,b-1}, \qvec(\tvec_{b}))$, $\xvec_2(\svec_{2,b}, \msg_{2,b-1}, \qvec(\tvec_{b}))$, $\xvec_3(\msg_{1,b-1},\msg_{2,b-1})$, $\wvec_{3,b}$, $\yvec_{3,b}$ by $\svec_{2}$,$\tvec$, $\qvec$, $\vvec_1$, $\vvec_2$, $\tilde{\xvec}_1$, $\xvec_2$, $\xvec_3$, $\wvec_{3}$, $\yvec_{3}$, respectively.
%
%
Note that the joint distribution 
obeys
\begin{align}
    & p(\tilde{\svec}_1, \svec_2, \tvec, \qvec, \vvec_1, \vvec_2, \tilde{\xvec}_1, \xvec_2, \xvec_3, \wvec_3, \yvec_3) \nonumber \\
		& \quad = \prod_{j=1}^{n}{ p(s_{2,j}, w_{3,j}, t_j) p(\tilde{s}_{1,j})} \times \nonumber \\
		& \mspace{80mu} p(v_{1,j}, v_{2,j}, x_{3,j}, q_j|\tilde{s}_{1,j},s_{2,j},w_{3,j}, t_j) \times \nonumber \\
    & \mspace{95mu} p(x_{2,j}, y_{3,j}|s_{2,j}, w_{3,j}, t_j, v_{1,j}, v_{2,j}, x_{3,j}, q_j) \times \nonumber \\
		& \mspace{110mu} p(\tilde{x}_{1,j}| v_{1,j}, v_{2,j},x_{3,j}, q_j,\tilde{s}_{1,j}).
\label{eq:RelayE21_Dist_CP}
\end{align}

Next, we use the assignments
\begin{align}
    & \zvec_1 = (\svec_{2}, \wvec_3, \tvec), \quad \zvec_2=\tilde{\svec}_{1}, \quad \Zvec_3=(\Vvec_1,\Vvec_2, \Xvec_3, \Qvec), \nonumber \\
		& \Zvec_4=\tilde{\Xvec}_1, \quad \Zvec_5 = (\Xvec_2,\Yvec_{3}).
\label{eq:Relay_E21_Zvals_CP}
\end{align}
\noindent Equation \eqref{eq:RelayE21_Dist_CP} shows that the assignments \eqref{eq:Relay_E21_Zvals_CP} satisfy the assumptions of \cite[Lemma, Appendix A]{Cover:80}.
Using this lemma we bound $\Pr \big\{ E_{21}^{r}(\tilde{\svec}_1) | (E_1^{r})^{c}, \tilde{\svec}_1 \in \styp(S_1|\svec_{2,b}, \wvec_{3,b}, \tvec_b) \big\}$
 as follows
\begin{align}
    & \Pr \big\{ E_{21}^{r}(\tilde{\svec}_1) | (E_1^{r})^{c}, \tilde{\svec}_1 \in \styp(S_1|\svec_{2,b}, \wvec_{3,b}, \tvec_b) \big\} \nonumber \\
    & \qquad \leq 2^{-n[I(X_1;Y_3|S_2,V_1,X_2,X_3,W_3,Q) -\epsilon_0]},
\label{eq:Relay_E21spec_errProb_CP}
\end{align}

\begin{figure*}[!b]
\normalsize
\setcounter{MYtempeqncnt}{\value{equation}}
\setcounter{equation}{16}

\vspace*{4pt}
\hrulefill
\vspace{-0.3cm}

\begin{align}
	& \mspace{-15mu} \Pr \big\{ E_{24}^{r}(\tsvec_1, \tsvec_2, \ttvec) | \msA_b, (E_1^{r})^{c} \big\} \nonumber \\
	& \quad \le \sum_{\tqvec \in \mQ^n}{\Pr \Big\{ \Qvec(\ttvec) = \tqvec \big| \msA_b, (E_1^{r})^{c} \Big\} \Pr \Big\{ \Qvec(\tvec_b) \ne \tqvec \big| \msA_b, (E_1^{r})^{c} \Big\} \times} \nonumber \\	
	    & \mspace{100mu} \Pr \Big\{ \Big( \tsvec_{1}, \tsvec_{2}, \ttvec, \tqvec, \Vvec_1(\msg_{1,b-1}), \Vvec_2(\msg_{2,b-1}), \Xvec_1(\tsvec_{1}, \msg_{1,b-1}, \tqvec), \Xvec_2(\tsvec_{2}, \msg_{2,b-1}, \tqvec),  \nonumber \\
    & \mspace{170mu} \Xvec_3(\msg_{1,b-1}, \msg_{2,b-1}), \wvec_{3,b}, \Yvec_{3,b} \Big) \in \styp \big| \msA_b, (E_1^{r})^{c}, \Qvec(\tvec_b) \ne \tqvec \Big\} \nonumber \\
		& \quad \le \sum_{\tqvec \in \styp(Q)}{\Pr \Big\{ \Qvec(\ttvec) = \tqvec \big| \msA_b, (E_1^{r})^{c} \Big\} }
\Pr \Big\{ \Big( \tsvec_{1}, \tsvec_{2}, \ttvec, \tqvec, \Vvec_1(\msg_{1,b-1}), \Vvec_2(\msg_{2,b-1}), \Xvec_1(\tsvec_{1}, \msg_{1,b-1}, \tqvec),   \nonumber \\
    & \mspace{170mu} \Xvec_2(\tsvec_{2}, \msg_{2,b-1}, \tqvec),\Xvec_3(\msg_{1,b-1}, \msg_{2,b-1}), \wvec_{3,b}, \Yvec_{3,b} \Big) \in \styp \big| \msA_b, (E_1^{r})^{c}, \Qvec(\tvec_b) \ne \tqvec \Big\}
        \label{eq:Relay_E24_expand_BasicTerm_CP},
\end{align}

\setcounter{equation}{\value{MYtempeqncnt}}

\end{figure*}

\noindent where $\eps_0 = 8\eps$, and we used
the fact that $T$ is a function of $S_2$, 
and the
Markov chains $V_2 - (S_2,V_1,X_2,$ $X_3,W_3,Q) - Y_3$, and 
$S_1 - (S_2,V_1,X_1,X_2,X_3,W_3,Q) - Y_3$.
Plugging \eqref{eq:Relay_E21spec_errProb_CP} into \eqref{eq:RelayE21_Bound} we have
\begin{align}
    & \mspace{-13mu} \Pr \big\{ E_{21}^{r} | (E_1^{r})^{c} \big\} \nonumber \\
    & \quad \leq \mspace{-12mu} \sum_{\substack{\tilde{\svec}_1 \neq \svec_{1,b}, \\ \tilde{\svec}_1 \in \styp(S_1|\svec_{2,b}, \wvec_{3,b}, \tvec_b)}}{\mspace{-50mu} 2^{-n[I(X_1;Y_3|S_2,V_1,X_2,X_3,W_3,Q) -\epsilon_0]}} \nonumber \\
%
    & \quad \leq 2^{n[H(S_1|S_2,W_3,T) - I(X_1;Y_3|S_2,V_1,X_2,X_3,W_3,Q) + 2\eps_0]} \nonumber \\
		& \quad  = 2^{n[H(S_1|S_2,W_3) - I(X_1;Y_3|S_2,V_1,X_2,X_3,W_3,Q) + 2\eps_0]},
\label{eq:Relay_E21_errProb_CP}
\end{align}

\noindent which can be bounded by $\epsilon$, for large enough $n$, as long as
\begin{align}
    &H(S_1|S_2,W_3) \nonumber \\
		& \qquad < I(X_1;Y_3|S_2,V_1,X_2,X_3,W_3,Q) - 2\eps_0.
\label{eq:Relay_E21_Cond_CP}
\end{align}

\noindent Following similar arguments as in \eqref{eq:Relay_E21SpecificDef_CP}--\eqref{eq:Relay_E21_errProb_CP},
we can also show that $\Pr \big\{ E_{22}^{r} |(E_1^{r})^{c} \big\}$ can be bounded by $\epsilon$, for large enough $n$, as long as
\begin{align}
    & H(S_2|S_1,W_3) \nonumber \\
		& \qquad < I(X_2;Y_3|S_1,V_2,X_1,X_3,W_3,Q) - 2\eps_0,
\label{eq:Relay_E22_Cond_CP}
\end{align}

	\noindent and $\Pr \big\{ E_{23}^{r} | (E_1^{r})^{c} \big\}$ can be bounded by $\epsilon$, for large enough $n$, as long as
\begin{align}
    &H(S_1,S_2|W_3,T) \nonumber \\
		& \qquad < I(X_1,X_2;Y_3|V_1,V_2,X_3,W_3,T,Q) - 2\eps_0.
\label{eq:Relay_E23_Cond_CP}
\end{align}

	To bound $\Pr \big\{ E_{24}^{r} | (E_1^{r})^{c} \big\}$ we first define the event $E_{24}^{r}(\tsvec_1, \tsvec_2, \ttvec)$ as follows
\begin{align}
    & E_{24}^{r}(\tsvec_1, \tsvec_2, \ttvec) \triangleq \nonumber \\
		& \quad \big\{ \Qvec(\ttvec) \ne \Qvec(\tvec_b), \big( \tsvec_{1}, \tsvec_{2}, \ttvec, \Qvec(\ttvec), \Vvec_1(\msg_{1,b-1}), \Vvec_2(\msg_{2,b-1}),   \nonumber \\
    & \qquad  \Xvec_1(\tsvec_{1}, \msg_{1,b-1}, \Qvec(\ttvec)), \Xvec_2(\tsvec_{2}, \msg_{2,b-1}, \Qvec(\ttvec)), \nonumber \\
		& \qquad \quad \Xvec_3(\msg_{1,b-1}, \msg_{2,b-1}), \wvec_{3,b}, \Yvec_{3,b} \big) \in \styp \big\}.
\label{eq:Relay_E24SpecificDef_CP}
\end{align}

\noindent Recalling that $\tsvec_{1} \ne \svec_{1,b}$, $\tsvec_2 \ne \svec_{2,b}$, $\ttvec \ne \tvec_b$, we have
\begin{align}
    & \Pr \big\{ E_{24}^{r} | (E_1^{r})^{c} \big\} \nonumber \\
		& \quad \leq \sum_{\substack{\tilde{\svec}_1 \neq \svec_{1,b}, \tsvec_2 \neq \svec_{2,b}, \ttvec \ne \tvec_b, \\ (\tsvec_1, \tsvec_2, \ttvec) \in \styp(S_1,S_2,T|\wvec_{3,b})}}{\mspace{-60mu}
    \Pr \big\{ E_{24}^{r}(\tsvec_1, \tsvec_2, \ttvec) | (E_1^{r})^{c} \big\}} .\label{eq:RelayE24_Bound}
\end{align}

%
%

\setcounter{equation}{17}

\noindent Let $\msA_b$ denote the event that $(\tsvec_1, \tsvec_2, \ttvec) \in \styp(S_1,S_2,T|\wvec_{3,b})$.
Note that if $\msA_b^c$ holds then $\Pr \Big\{ E_{24}^{r}(\tsvec_1, \tsvec_2, \ttvec)\Big| (E_1^r)^c , \msA_b^c \Big\} = 0$.
Hence, we can~write \eqref{eq:Relay_E24_expand_BasicTerm_CP} at the bottom of the page, 
\noindent where \eqref{eq:Relay_E24_expand_BasicTerm_CP} follows from \cite[Thm. 6.7]{YeungBook}.
%
In the following we upper bound the summands in \eqref{eq:Relay_E24_expand_BasicTerm_CP}
%
independently of $\tsvec_{1}, \tsvec_{2}, \ttvec$, and $\tqvec$.
To reduce clutter let us denote $\vvec_1(\msg_{1,b-1})$, $\vvec_2(\msg_{2,b-1})$, $\xvec_1(\tilde{\svec}_{1}, \msg_{1,b-1}, \tqvec)$, $\xvec_2(\tsvec_{2}, \msg_{2,b-1}, \tqvec)$, $\xvec_3(\msg_{1,b-1},\msg_{2,b-1})$, $\wvec_{3,b}$, $\yvec_{3,b}$ by $\vvec_1$, $\vvec_2$, $\tilde{\xvec}_1$, $\txvec_2$, $\xvec_3$, $\wvec_{3}$,$\yvec_{3}$, respectively.
The joint distribution
obeys
\begin{align}
    & p(\tilde{\svec}_1, \tsvec_2, \ttvec, \tqvec, \vvec_1, \vvec_2, \txvec_1, \txvec_2,  \xvec_3, \wvec_3, \yvec_3) \nonumber \\
   & \mspace{10mu} = \prod_{j=1}^{n}{p(w_{3,j}) p(\tilde{s}_{1,j}, \tilde{s}_{2,j}, \tilde{t}_j, \tilde{q}_j)} \times \nonumber \\
	& \mspace{70mu} p(v_{1,j}, v_{2,j}, x_{3,j} |\tilde{s}_{1,j},\tilde{s}_{2,j},w_{3,j}, \tilde{t}_j, \tilde{q}_j) \times \nonumber \\
    & \mspace{85mu} p(\tilde{x}_{1,j}, \tilde{x}_{2,j}| \tilde{s}_{1,j}, \tilde{s}_{2,j},\tilde{t}_j, \tilde{q}_j,v_{1,j}, v_{2,j},x_{3,j}) \times \nonumber \\
		& \mspace{100mu} p(y_{3,j}|w_{3,j}, v_{1,j}, v_{2,j}, x_{3,j}).
\label{eq:RelayE24_Dist_CP}
\end{align}	

\noindent
Moreover, note that the independence of $Q$ from $(S_1,S_2,T)$, and conditioning on $\msA_b$, implies that
for $n$ large enough $(\tsvec_1, \tsvec_2, \ttvec, \tqvec) \in \stypt(S_1,S_2,T,Q|\wvec_{3,b})$. Hence, we can use \cite[Lemma, Appendix A]{Cover:80} with the following assignments:
\begin{align}
    & \zvec_1= \wvec_3, \quad \zvec_2=(\tsvec_{1}, \tsvec_2, \ttvec, \tqvec), \quad \Zvec_3=(\Xvec_3,\Vvec_1,\Vvec_2), \nonumber \\
		& \Zvec_4=(\tXvec_1, \tXvec_2), \quad \Zvec_5 = \Yvec_{3},
\label{eq:Relay_E24_Zvals_CP}
\end{align}

\noindent to bound
%
%
\begin{align}
 \Pr \Big\{ \Big( & \tsvec_{1}, \tsvec_{2}, \ttvec, \tqvec, \Vvec_1(\msg_{1,b-1}), \Vvec_2(\msg_{2,b-1}), \Xvec_1(\tsvec_{1}, \msg_{1,b-1}, \tqvec), \nonumber \\
	&\Xvec_2(\tsvec_{2}, \msg_{2,b-1}, \tqvec),
        \Xvec_3(\msg_{1,b-1}, \msg_{2,b-1}), \wvec_{3,b}, \Yvec_{3,b} \Big) \nonumber \\
				& \in \stypt \Big| (E_1^{r})^{c}, \Qvec(\tvec_b) \ne \tqvec, \nonumber \\
				& \mspace{12mu} (\tsvec_1, \tsvec_2, \ttvec, \tqvec) \in \stypt(S_1,S_2,T,Q|\wvec_{3,b})\Big\} \nonumber \\
    & \mspace{20mu} \leq 2^{-n[I(S_1,S_2,T,Q,X_1,X_2;Y_3|V_1,V_2,X_3,W_3) -\epsilon_0]} \nonumber \\
    & \mspace{20mu} \stackrel{(a)}{=} 2^{-n[I(X_1,X_2;Y_3|V_1,V_2,X_3,W_3) -\epsilon_0]},
\label{eq:Relay_E24spec_errProb_Partial_CP}
\end{align}

\noindent where $(a)$ follows from the Markov chain $(S_1,S_2,T,Q) - (X_1,X_2,V_1,V_2,X_3,W_3) - Y_3$. From  \eqref{eq:Relay_E24spec_errProb_Partial_CP} and \eqref{eq:Relay_E24_expand_BasicTerm_CP} we obtain
\begin{align}
	& \mspace{-8mu} \Pr \big\{ E_{24}^{r}(\tsvec_1, \tsvec_2, \ttvec) | \msA_b, (E_1^{r})^{c}  \big\} \nonumber \\
	& \mspace{1mu} \le \mspace{-40mu} \sum_{ \mspace{35mu} \tqvec \in \styp(Q)}{\mspace{-38mu} 2^{-n[I(X_1,X_2;Y_3|V_1,V_2,X_3,W_3) -\epsilon_0]} 2^{-n[H(Q) - \epsilon_1]}},
\end{align}

\noindent and by using the bound ${\dst \left| \styp(Q) \right| \le 2^{n[H(Q) + \eps_1]}}$, \cite[Thm. 6.2]{YeungBook}, we have that
\begin{align}
	&\Pr \big\{ E_{24}^{r}(\tsvec_1, \tsvec_2, \ttvec) | \msA_b, (E_1^{r})^{c} \big\} \nonumber \\
	& \quad \le 2^{-n[I(X_1,X_2;Y_3|V_1,V_2,X_3,W_3) -\epsilon_0 - 2\eps_1]}. \label{eq:Relay_E24spec_errProb_CP}
\end{align}

\noindent Plugging \eqref{eq:Relay_E24spec_errProb_CP} into \eqref{eq:RelayE24_Bound} we have
\begin{align}
    & \Pr \big\{ E_{24}^{r} | (E_1^{r})^{c} \big\} \nonumber \\
%
    & \qquad  \leq \left| \styp(S_1,S_2,T|\wvec_{3,b}) \right|  \times \nonumber \\
		& \qquad \qquad \quad 2^{-n[I(X_1,X_2;Y_3|V_1,V_2,X_3,W_3) -\epsilon_0 - 2\eps_1]} \nonumber \\
%
%
		& \qquad \le  2^{n[H(S_1,S_2|W_3) - I(X_1,X_2;Y_3|V_1,V_2,X_3,W_3) + 4\eps_0]},
    \label{eq:Relay_E24_errProb_CP}
\end{align}

\noindent which can be bounded by $\epsilon$, for large enough $n$, if
\begin{align}
    \mspace{-7mu} H(S_1,S_2|W_3) \mspace{-2mu} < \mspace{-2mu} I(X_1,X_2;Y_3|V_1,V_2,X_3,W_3) \mspace{-1mu} - \mspace{-1mu} 4\eps_0.
\label{eq:Relay_E24_Cond_CP}
\end{align}

Lastly, to bound $\Pr \big\{ E_{25}^{r} | (E_1^{r})^{c} \big\}$ we first define the event $E_{25}^{r}(\tsvec_1, \tsvec_2, \ttvec)$ as follows
\begin{align}
    & E_{25}^{r}(\tsvec_1, \tsvec_2, \ttvec) \triangleq \nonumber \\
		& \quad \big\{ \Qvec(\ttvec) = \Qvec(\tvec_b),  \big( \tsvec_{1}, \tsvec_{2}, \ttvec, \Qvec(\ttvec), \Vvec_1(\msg_{1,b-1}), \Vvec_2(\msg_{2,b-1}), \nonumber \\
		& \qquad \Xvec_1(\tsvec_{1}, \msg_{1,b-1}, \Qvec(\ttvec)), \Xvec_2(\tsvec_{2}, \msg_{2,b-1},\Qvec(\ttvec)), \nonumber \\
		& \qquad \quad \Xvec_3(\msg_{1,b-1}, \msg_{2,b-1}), \wvec_{3,b}, \Yvec_{3,b} \big) \in \styp \big\}.
\label{eq:Relay_E25SpecificDef_CP}
\end{align}

\noindent Recalling that $\tsvec_1 \ne \svec_{1,b}, \tsvec_2 \ne \svec_{2,b}, \ttvec \ne \tvec_b$, we have
\begin{align}
    & \Pr \big\{ E_{25}^{r} | (E_1^{r})^{c} \big\} \nonumber \\
		& \qquad \leq \mspace{-20mu} \sum_{\substack{\tilde{\svec}_1 \neq \svec_{1,b}, \tsvec_2 \neq \svec_{2,b}, \ttvec \ne \tvec_b, \\ (\tsvec_1, \tsvec_2, \ttvec) \in \styp(S_1,S_2,T|\wvec_{3,b})}}{\mspace{-60mu}
    \Pr \big\{ E_{25}^{r}(\tsvec_1, \tsvec_2, \ttvec) | (E_1^{r})^{c} \big\}}. \label{eq:RelayE25_Bound}
\end{align}

\noindent Then we have
\begin{align}
	& \mspace{-10mu} \Pr \big\{ E_{25}^{r}(\tsvec_1, \tsvec_2, \ttvec) | \msA_b, (E_1^{r})^{c} \big\} \nonumber \\
%
%
		& \mspace{10mu} = \mspace{-25mu} \sum_{ \mspace{25mu} \bqvec \in \styp(Q)}{ \mspace{-25mu} \Pr \Big\{ \Qvec(\ttvec) = \bqvec \big| \msA_b, (E_1^{r})^{c} \Big\}} \times \nonumber \\
		& \mspace{80mu} \Pr \Big\{ \Qvec(\tvec_b) = \bqvec \big| \msA_b, (E_1^{r})^{c} \Big\} \times \nonumber \\	
	    & \mspace{80mu} \Pr \Big\{ \Big( \tsvec_{1}, \tsvec_{2}, \ttvec, \bqvec, \Vvec_1(\msg_{1,b-1}), \Vvec_2(\msg_{2,b-1}), \nonumber \\
			& \mspace{125mu} \Xvec_1(\tsvec_{1}, \msg_{1,b-1}, \bqvec), \Xvec_2(\tsvec_{2}, \msg_{2,b-1}, \bqvec),  \nonumber \\
    & \mspace{125mu} \Xvec_3(\msg_{1,b-1}, \msg_{2,b-1}), \wvec_{3,b}, \Yvec_{3,b} \Big) \nonumber \\
		& \mspace{125mu} \in \styp \big| \msA_b, (E_1^{r})^{c}, \Qvec(\tvec_b) = \bqvec \Big\} \label{eq:Relay_E25_expand_BasicTerm_CP},
\end{align}

\noindent where \eqref{eq:Relay_E25_expand_BasicTerm_CP} follows from the same argument leading to \eqref{eq:Relay_E24_expand_BasicTerm_CP}. In the following we upper bound the summands in \eqref{eq:Relay_E25_expand_BasicTerm_CP}
%
independently of $\tsvec_{1}, \tsvec_{2}, \ttvec$, and $\bqvec$.
Let us denote $\vvec_1(\msg_{1,b-1})$, $\vvec_2(\msg_{2,b-1})$, $\xvec_1(\tilde{\svec}_{1}, \msg_{1,b-1}, \bqvec)$, $\xvec_2(\tsvec_{2}, \msg_{2,b-1}, \bqvec)$, $\xvec_3(\msg_{1,b-1},\msg_{2,b-1})$, $\wvec_{3,b}$, $\yvec_{3,b}$ by $\vvec_1, \vvec_2, \tilde{\xvec}_1, \txvec_2, \xvec_3, \wvec_{3}, \yvec_{3}$, respectively.
Note that the joint distribution~obeys
\begin{align}
    & \mspace{-10mu} p(\tilde{\svec}_1, \tsvec_2, \ttvec, \bqvec, \vvec_1, \vvec_2, \txvec_1, \txvec_2, \xvec_3, \wvec_3, \yvec_3) \nonumber \\
		& \mspace{15mu}  = \prod_{j=1}^{n}{p(\tilde{s}_{1,j}, \tilde{s}_{2,j}, \tilde{t}_j, \bar{q}_j,w_{3,j})} \times \nonumber \\
		& \mspace{75mu} p(v_{1,j}, v_{2,j}, x_{3,j} |\tilde{s}_{1,j},\tilde{s}_{2,j},w_{3,j}, \tilde{t}_j, \bar{q}_j) \times \nonumber \\
    & \mspace{85mu} p(\tilde{x}_{1,j}, \tilde{x}_{2,j}| \tilde{s}_{1,j}, \tilde{s}_{2,j},\tilde{t}_j, \bar{q}_j,v_{1,j}, v_{2,j},x_{3,j}) \times \nonumber \\
		& \mspace{95mu} p(y_{3,j}|w_{3,j}, v_{1,j}, v_{2,j}, x_{3,j}, \bar{q}_j).
\label{eq:RelayE25_Dist_CP}
\end{align}	

\noindent 
Similarly to the analysis for  $E_{24}^{r}(\tsvec_1, \tsvec_2, \ttvec)$ we have
that $(\tsvec_1, \tsvec_2, \ttvec, \bqvec) \in \stypt(S_1,S_2,T,Q|\wvec_{3,b})$. Hence, we can use \cite[Lemma, Appendix A]{Cover:80} with the following assignments:
%
%
\begin{align}
    & \zvec_1= (\wvec_3, \bqvec), \quad \zvec_2=(\tsvec_{1}, \tsvec_2, \ttvec, \bqvec), \quad \Zvec_3=(\Xvec_3,\Vvec_1,\Vvec_2), \nonumber \\
		& \Zvec_4=(\tXvec_1, \tXvec_2), \quad \Zvec_5 = \Yvec_{3}.
\label{eq:Relay_E25_Zvals_CP}
\end{align}

\noindent Then we get the following bound:
\begin{align}
   \Pr \Big\{ \Big( & \tsvec_{1}, \tsvec_{2}, \ttvec, \bqvec, \Vvec_1(\msg_{1,b-1}), \Vvec_2(\msg_{2,b-1}), \Xvec_1(\tsvec_{1}, \msg_{1,b-1}, \bqvec), \nonumber \\
	& \Xvec_2(\tsvec_{2}, \msg_{2,b-1}, \bqvec), \Xvec_3(\msg_{1,b-1}, \msg_{2,b-1}), \wvec_{3,b}, \Yvec_{3,b} \Big) \nonumber \\
	& \mspace{8mu} \in \stypt \Big| (E_1^{r})^{c}, \Qvec(\tvec_b) = \bqvec, \nonumber \\
	& \mspace{20mu} (\tsvec_1, \tsvec_2, \ttvec, \bqvec) \in \stypt(S_1,S_2,T,Q|\wvec_{3,b}) \Big\} \nonumber \\
%
    & \mspace{40mu} \stackrel{(a)}{\le} 2^{-n[I(X_1,X_2;Y_3|V_1,V_2,X_3,W_3,Q) -\epsilon_0]},
\label{eq:Relay_E25spec_errProb_Partial_CP}
\end{align}

\noindent where $(a)$ follows from the Markov chain $(S_1,S_2,T) - (X_1,X_2,V_1,V_2,X_3,W_3) - Y_3$. From \eqref{eq:Relay_E25spec_errProb_Partial_CP} and \eqref{eq:Relay_E25_expand_BasicTerm_CP} we have
\begin{align}
	 & \Pr \big\{ E_{25}^{r}(\tsvec_1, \tsvec_2, \ttvec) | \msA_b, (E_1^{r})^{c} \big\} \nonumber \\
	& \qquad \le \sum_{\bqvec \in \styp(Q)}{\mspace{-20mu} 2^{-n[I(X_1,X_2;Y_3|V_1,V_2,X_3,W_3,Q) -\epsilon_0]} \times} \nonumber \\
	& \mspace{125mu} \Pr \Big\{ \Qvec(\ttvec) = \bqvec \big| \msA_b, (E_1^{r})^{c} \Big\} \times \nonumber \\
	& \mspace{150mu} \Pr \Big\{ \Qvec(\tvec_b) = \bqvec \big| \msA_b, (E_1^{r})^{c} \Big\}.
\end{align}

\noindent However, for $\bqvec \in \styp(Q)$ the following holds
\begin{align*}
	\Pr \Big\{ \Qvec(\ttvec) = \bqvec \big| \msA_b, (E_1^{r})^{c} \Big\} & \le 2^{-n[H(Q) - \eps_1]}, \\
	\Pr \Big\{ \Qvec(\tvec_b) = \bqvec \big| \msA_b, (E_1^{r})^{c} \Big\} & \le 2^{-n[H(Q) - \eps_1]}. 
\end{align*}

\noindent Hence,
%
using the fact that ${\dst \left| \styp(Q) \right| \le 2^{n[H(Q) + \eps_1]}}$, we have that
\begin{align}
	& \Pr \big\{ E_{25}^{r}(\tsvec_1, \tsvec_2, \ttvec) | \msA_b, (E_1^{r})^{c} \big\} \nonumber \\
	& \qquad \le 2^{-n[I(X_1,X_2;Y_3|V_1,V_2,X_3,W_3,Q) + H(Q) -\epsilon_0 - 3 \eps_1]}. \label{eq:Relay_E25spec_errProb_CP}
\end{align}

\noindent Finally, plugging \eqref{eq:Relay_E25spec_errProb_CP} into \eqref{eq:RelayE25_Bound} we have
\begin{align}
    & \Pr \big\{ E_{25}^{r} | (E_1^{r})^{c} \big\} \nonumber \\
    & \quad \leq \left| \styp(S_1,S_2,T|\wvec_{3,b}) \right|  \times \nonumber \\
		& \qquad \qquad 2^{-n[I(X_1,X_2;Y_3|V_1,V_2,X_3,W_3,Q) + H(Q) -\epsilon_0 - 3 \eps_1]} \nonumber \\
%
%
		& \quad \le 2^{n[H(S_1,S_2|W_3) - I(X_1,X_2;Y_3|V_1,V_2,X_3,W_3,Q) - H(Q)  + 5\eps_0]},
    \label{eq:Relay_E25_errProb_CP}
\end{align}

\noindent which can be bounded by $\epsilon$, for large enough $n$, as long as
\begin{align}
    & \mspace{-8mu} H(S_1,S_2|W_3) \nonumber \\
		& \mspace{2mu} < I(X_1,X_2;Y_3|V_1,V_2,X_3,W_3,Q) + H(Q) - 5\eps_0.
\label{eq:Relay_E25_Cond_CP}
\end{align}

\noindent However, condition \eqref{eq:Relay_E25_Cond_CP} is redundant since it is dominated by condition \eqref{eq:Relay_E24_Cond_CP},
%
	hence, we conclude that if conditions \eqref{bnd:Joint_rly_S1_CP}-\eqref{bnd:Joint_rly_S1S2_2_CP} hold, then for large enough $n$,
\begin{eqnarray}
    \Pr \big\{ E_2^{r} | (E_1^{r})^{c} \big\} \leq \sum_{j=1}^{5}{\Pr \big\{ E_{2j}^{r} | (E_1^{r})^{c} \big\}} \leq 5\epsilon.
\label{eq:RelayJntE2Rbound2_CP}
\end{eqnarray}

\noindent Combining equations \eqref{eq:RelayDecErrProbDef_CP}, \eqref{eq:RelayDecErrProb1_CP} and \eqref{eq:RelayJntE2Rbound2_CP} yields
\begin{equation}
    \bar{P}_{r}^{(n)} \leq \Pr \big\{ E_2^{r} | (E_1^{r})^{c} \big\} + 2\epsilon \leq 7\epsilon.
\label{eq:RelayDecErrProbRes_CP}
\end{equation}

\noindent Next the destination error probability analysis is derived.

\begin{figure*}[!b]
\normalsize
\setcounter{MYtempeqncnt}{\value{equation}}
\setcounter{equation}{36}

\vspace*{4pt}
\hrulefill
\vspace{-0.3cm}

\begin{align}
    \bar{P}_{d,ch}^{(n)}
    & \triangleq \sum_{(\svec_{1,b+1},\svec_{2,b+1}) \in \mS_1^n \times \mS_2^n}{\mspace{-30mu} p(\svec_{1,b+1},\svec_{2,b+1})} \sum_{(\msg_{1,b},\msg_{2,b}) \in \msgCal_1 \times \msgCal_2}{\mspace{-24mu} p(\msg_{1,b},\msg_{2,b})}
        \Pr \big\{ E^{d}_{ch}(\msg_{1,b},\msg_{2,b};\svec_{1,b+1},\svec_{2,b+1}) \big\} \nonumber \\
		& \stackrel{(a)}{\le} \eps + \mspace{-12mu} \sum_{ (\svec_{1,b+1},\svec_{2,b+1},\wvec_{b+1}, \tvec_{b+1}) \in \styp(S_1,S_2,W,T)}{\mspace{-100mu} p(\svec_{1,b+1},\svec_{2,b+1},\wvec_{b+1})} \times \nonumber \\
		& \qquad \qquad\qquad\qquad\qquad \sum_{(\msg_{1,b},\msg_{2,b}) \in \msgCal_1 \times \msgCal_2}{\mspace{-24mu} p(\msg_{1,b},\msg_{2,b})} \Pr \big\{ E^{d}_{ch}(\msg_{1,b},\msg_{2,b};\svec_{1,b+1},\svec_{2,b+1}) | \msD_{b+1} \big\}. \label{eq:DestDecErrProbDefFlip_CP}
\end{align}

\setcounter{equation}{\value{MYtempeqncnt}}

\end{figure*}

{\bf Channel decoder:} Let $E^{d}_{ch}(\msg_{1,b},\msg_{2,b};\svec_{1,b+1},\svec_{2,b+1})$ denote the  channel decoding error event
for decoding $(u_{1,b},u_{2,b})$ at the destination at block $b$, assuming $(\svec_{1,b+1},\svec_{2,b+1})$ is available at the destination,
namely, the event that $(\hmsg_{1,b},\hmsg_{2,b}) \ne (\msg_{1,b},\msg_{2,b})$. Let $\tvec_{b+1}=h_1(\svec_{1,b+1})=h_2(\svec_{2,b+1})$.
The average probability of channel decoding error at the destination at block $b$, $\bar{P}_{d,ch}^{(n)}$, is defined in \eqref{eq:DestDecErrProbDefFlip_CP} at the bottom of the page,
%
\noindent where in step (a) leading to \eqref{eq:DestDecErrProbDefFlip_CP} we apply similar reasoning as \cite[Eq. (16)]{Cover:80}. In the following we show that the inner sum in \eqref{eq:DestDecErrProbDefFlip_CP} can be upper bounded independently of $(\svec_{1,b+1},\svec_{2,b+1})$.
Assuming correct decoding at block $b+1$ (hence $(\svec_{1,b+1},\svec_{2,b+1})$ are available at the destination), we now define the following events:
\vspace{-0.25cm}
\begin{align*}
    E_1^{d} \triangleq \big\{ &(\svec_{1,b+1}, \svec_{2,b+1}, \tvec_{b+1}, \Qvec(\tvec_{b+1}),  \\
		& \quad \Vvec_1(\msg_{1,b}), \Vvec_2(\msg_{2,b}), \Xvec_1(\svec_{1,b+1}, \msg_{1,b}, \Qvec(\tvec_{b+1})), \\
		& \qquad \Xvec_2(\svec_{2,b+1}, \msg_{2,b}, \Qvec(\tvec_{b+1})), \Xvec_3(\msg_{1,b}, \msg_{2,b}), \\
		& \qquad \quad \wvec_{b+1}, \Yvec_{b+1}) \notin \styp \big\},  \\
    E_2^{d} \triangleq \big\{ & \exists \hmsg_{1} \in \msgCal_1:  \hmsg_{1} \neq \msg_{1,b}, \\
		& \quad (\svec_{1,b+1}, \svec_{2,b+1}, \tvec_{b+1}, \Qvec(\tvec_{b+1}), \\
		& \qquad \Vvec_1(\hmsg_1), \Vvec_2(\msg_{2,b}), \Xvec_1(\svec_{1,b+1}, \hmsg_1, \Qvec(\tvec_{b+1})), \\
		& \qquad \quad \Xvec_2(\svec_{2,b+1}, \msg_{2,b}, \Qvec(\tvec_{b+1})), \Xvec_3(\hmsg_1, \msg_{2,b}), \\
		& \qquad \qquad \wvec_{b+1}, \Yvec_{b+1})\in \styp \big\},   \\
    E_3^d \triangleq \big\{ & \exists \hmsg_{2} \in \msgCal_2:  \hmsg_{2} \neq \msg_{2,b}, \\
		& \quad (\svec_{1,b+1}, \svec_{2,b+1}, \tvec_{b+1}, \Qvec(\tvec_{b+1}), \\
		& \qquad \Vvec_1(\msg_{1,b}), \Vvec_2(\hmsg_2), \Xvec_1(\svec_{1,b+1}, \msg_{1,b}, \Qvec(\tvec_{b+1})), \\
		& \qquad \quad \Xvec_2(\svec_{2,b+1}, \hmsg_2, \Qvec(\tvec_{b+1})), \Xvec_3(\msg_{1,b}, \hmsg_2), \\
		& \qquad \qquad \wvec_{b+1}, \Yvec_{b+1}) \in \styp \big\},   \\
    E_4^d \triangleq \big\{ & \exists (\hmsg_{1},  \hmsg_{2}) \in \msgCal_1 \times \msgCal_2:  \hmsg_{1} \neq \msg_{1,b}, ,  \hmsg_{2} \neq \msg_{2,b}, \\
		& \quad (\svec_{1,b+1}, \svec_{2,b+1}, \tvec_{b+1}, \Qvec(\tvec_{b+1}), \\
    & \qquad \Vvec_1(\hmsg_1), \Vvec_2(\hmsg_2), \Xvec_1(\svec_{1,b+1}, \hmsg_1, \Qvec(\tvec_{b+1})),  \\
		& \qquad \quad \Xvec_2(\svec_{2,b+1}, \hmsg_2, \Qvec(\tvec_{b+1})), \Xvec_3(\hmsg_1, \hmsg_2), \\
		& \qquad \qquad \wvec_{b+1}, \Yvec_{b+1})\in \styp \big\}.
\end{align*}

\setcounter{equation}{37}

\vspace{-0.15cm}
\noindent The average probability of error for decoding $(\msg_{1,b},\msg_{2,b})$ at the destination at block $b$, for fixed $(\svec_{1,b+1}, \svec_{2,b+1})$, subject to the event
$\msD_{b+1}$, 
is then upper bounded by
\vspace{-0.25cm}
\begin{align}
    & \Pr \left\{ E^{d}_{ch}(\msg_{1,b},\msg_{2,b};\svec_{1,b+1},\svec_{2,b+1}) | \msD_{b+1} \right\} \nonumber \\
		& \quad \leq  \Pr \big\{ E^{d}_{1} | \msD_{b+1} \big\} + \sum_{j=2}^{4}{\Pr \big\{ E^{d}_{j} | (E_1^{d})^{c} \big\}},
\label{eq:DestJntChanDecErrProb_CP}
\end{align}

\vspace{-0.25cm}
\noindent where \eqref{eq:DestJntChanDecErrProb_CP} follows from the union bound.
From the AEP \cite[Ch. 5.1]{YeungBook}, for sufficiently large $n$, $\Pr \{ E^{d}_1 | \msD_{b+1} \}$ can be upper bounded by $\epsilon$ for $n$ large enough. Let $\eps_0$ be a positive number such that $\eps_0 > \eps$ and $\eps_0 \rightarrow 0$ as $\eps \rightarrow 0$.
To bound $\Pr \left\{ E_2^d(\hmsg_1) | (E_1^{d})^{c} \right\}$ we first define the event $E_2^d(\hmsg_1)$ as follows
\begin{align}
    & E_2^d(\hmsg_1) \triangleq \nonumber \\
		& \quad \big\{ \big(\svec_{1,b+1}, \svec_{2,b+1}, \tvec_{b+1}, \Qvec(\tvec_{b+1}), \Vvec_1(\hmsg_1), \Vvec_2(\msg_{2,b}), \nonumber \\
		& \qquad \Xvec_1(\svec_{1,b+1}, \hmsg_1, \Qvec(\tvec_{b+1})), \Xvec_2(\svec_{2,b+1}, \msg_{2,b}, \Qvec(\tvec_{b+1})), \nonumber \\
		& \qquad \quad \Xvec_3(\hmsg_1, \msg_{2,b}), \wvec_{b+1}, \Yvec_{b+1} \big) \in \styp \big\}.
\label{eq:DestJntChanE2_msg1_def_CP}
\end{align}

\noindent Recalling that $\hmsg_1 \neq \msg_{1,b}$, we can bound
\begin{align}
	\Pr \left\{ E_2^d | (E_1^{d})^{c} \right\}  \leq  \sum_{\hmsg_1 \in \msgCal_1 , \hmsg_1 \neq \msg_{1,b}}{\mspace{-24mu}\Pr \left\{E_2^d(\hmsg_1) | (E_1^{d})^{c}\right\}}. \label{eq:DestJntE2R_basic_Bound_CP}
\end{align}

\noindent Using \cite[Thm. 14.2.3]{cover-thomas:it-book}, $\Pr \LL{\{} E_2^d(\hmsg_1) | (E_1^{d})^{c}\RR{\}}$ can be bounded by
\begin{align}
    & \Pr \LL{\{} E_2^d(\hmsg_1) | (E_1^{d})^{c} \RR{\}} \nonumber \\
		& \qquad \leq 2^{-n[I(V_1,X_1,X_3;Y|S_1,S_2,T,V_2,X_2,W,Q) - \epsilon_0]} \nonumber \\
%
    & \qquad \stackrel{(a)}{=}  2^{-n[I(X_1,X_3;Y|S_1,V_2,X_2,Q) - \epsilon_0]},    \label{eq:DestJntE2R_msg1_Bound_CP}
\end{align}

\noindent where (a) follows from the Markov chains $V_1 - (X_1,X_2,X_3,$ $S_1,S_2,W,T,Q,V_2) - Y$ and $(S_2,T,W) - (X_1,X_2,X_3,S_1,$ $V_2,X_2,Q) - Y$.
Plugging \eqref{eq:DestJntE2R_msg1_Bound_CP} into \eqref{eq:DestJntE2R_basic_Bound_CP} we have
\begin{align}
    \mspace{-8mu} \Pr \left\{ E_2^d | (E_1^{d})^{c} \right\}
    & \leq 2^{n[R_1 - (I(X_1,X_3;Y|S_1,X_2,V_2,Q) - \epsilon_0)]},
\label{eq:DestJntE2RBound_CP}
\end{align}

\noindent which can be bounded by $\epsilon$, for large enough $n$, as long as
\begin{equation}
    R_1 < I(X_1,X_3;Y|S_1,X_2,V_2,Q) - \epsilon_0.
\label{eq:DestJntE2RDecCondition_CP}
\end{equation}

Following similar arguments as in \eqref{eq:DestJntChanE2_msg1_def_CP}--\eqref{eq:DestJntE2RDecCondition_CP}, we can also show that $\Pr \left\{ E_3^d | (E_1^{d})^{c} \right\}$ can be bounded by $\epsilon$,
for large enough $n$, as long as
\begin{equation}
    R_2 < I(X_2,X_3;Y|S_2,X_1,V_1,Q) - \epsilon_0,
\label{eq:DestJntE3RDecCondition_CP}
\end{equation}

\noindent and $\Pr \left\{ E_4^d | (E_1^{d})^{c} \right\}$ can be bounded by $\epsilon$, for large enough $n$, as long as
\begin{equation}
    R_1 + R_2 < I(X_1,X_2,X_3;Y|S_1,S_2,Q) - \epsilon_0.
\label{eq:DestJntE4RDecCondition_CP}
\end{equation}

Hence, if conditions \eqref{eq:DestJntE2RDecCondition_CP}--\eqref{eq:DestJntE4RDecCondition_CP} hold, for large enough $n$, $\bar{P}_{d,ch}^{(n)} \leq 5\epsilon$.

\textbf{Source decoder:} From the SW theorem \cite{SW:73} it follows that, given correct decoding of $(u_{1,b},u_{2,b})$, the average probability of error in decoding $(\svec_{1,b},\svec_{2,b})$
at the destination can be made arbitrarily small for sufficiently large $n$, as long as
\begin{subequations} \label{eq:DestJntSourceCondE678_CP}
\begin{eqnarray}
    H(S_1|S_2,W) + \epsilon_0 &<& R_1 , \label{eq:DestJntSourceCondE6_CP} \\
    H(S_2|S_1,W) + \epsilon_0 &<& R_2 , \label{eq:DestJntSourceCondE7_CP} \\
    H(S_1,S_2|W) + \epsilon_0 &<& R_1 + R_2 . \label{eq:DestJntSourceCondE8_CP}
\end{eqnarray}
\end{subequations}

Combining conditions \eqref{eq:DestJntE2RDecCondition_CP}--\eqref{eq:DestJntE4RDecCondition_CP} with conditions \eqref{eq:DestJntSourceCondE678_CP} yields the destination decoding constrains \eqref{bnd:Joint_dst_S1_CP}--\eqref{bnd:Joint_dst_S1S2_CP} in \Thmref{thm:jointCond_CP}.


\vspace{-0.2cm}
\section{Proof of Theorem \thmref{thm:jointCondFlip_CP}}	\label{annex:jointProofFlip_CP}

Fix a distribution $p(s_1,s_2,w_3,w)p(q)p(x_1|s_1,q)p(x_2|s_2,q)$ $p(x_3|s_1,s_2,q)p(y_3,y|x_1,x_2,x_3)$.

\subsection{Codebook construction}
For $i=1,2$, assign every $\svec_i \in \mS_i^n$ to one of $2^{nR_i}$ bins independently according to a uniform distribution on $\msgCal_i \triangleq \{1,2,\dots,2^{nR_i}\}$. Denote this assignment by $f_i, i=1,2$.

For each $\tvec \in \mT^n$ generate one $n$-length codeword $\qvec(\tvec)$ by choosing the letters $q_k$ independently with distribution $p_{Q}(q_{k})$, for $k = 1,2,\dots, n$.
For each pair $(\msg_i, \svec_i)\in \msgCal_i \times \mS_i^n, i=1,2$,  set $\tvec = h_i(\svec_i)$, and generate one $n$-length codeword $\xvec_i(\msg_i, \svec_i,\qvec(\tvec)), \svec_i \in \mS_i^n, \qvec \in \mQ^n$,
by choosing the letters $x_{i,k}(\msg_i, \svec_i, \qvec(\tvec))$ independently with distribution $p_{X_i|S_i,Q}(x_{i,k}|s_{i,k},q_k(\tvec))$ for $k = 1,2,\dots,n$.
Finally, generate one length-$n$ relay codeword $\xvec_3(\svec_1,\svec_2, \qvec(\tvec))$ for each pair $(\svec_1,\svec_2) \in \mS_1^n \times \mS_2^n$ by choosing $x_{3,k}(\svec_1,\svec_2,\qvec(\tvec))$
independently with distribution $p_{X_3|S_1,S_2,Q}(x_{3,k}|s_{1,k},s_{2,k},q_k(\tvec))$ for $k = 1,2,\dots,n$.

\subsection{Encoding}
    Consider the  sequences $s^{Bn}_{i,1} \in \mS^{Bn}_i, i=1,2$, $w_{3,1}^{Bn}\in\mW_3^{Bn}$, and $w^{Bn}\in\mW^{Bn}$, all of length $Bn$.
    Partition each sequence into $B$ length-$n$ subsequences, $\svec_{i,b}$, $i=1,2$, $\wvec_{3,b}$, and $\wvec_b$,  $b=1,2,\dots,B$.
A total of $Bn$ source samples are transmitted in $B+1$ blocks of $n$ channel symbols each.
Let $(\avec_1, \avec_2) \in \mS_1^n \times \mS_2^n$ be two sequences generated i.i.d according to $p(\avec_1,\avec_2) = \prod_{k=1}^{n}{p_{S_1,S_2}(a_{1,k}, a_{2,k})}$. These sequences are known to all nodes.
At block $1$, source terminal $i, i=1,2$, observes $\svec_{i,1}$, finds its corresponding bin index $\msg_{i,1} = f_i(\svec_{i,1})\in \msgCal_i$, and transmits the channel codeword $\xvec_i(\msg_{i,1}, \avec_i, \qvec(h_i(\avec_i)))$.
At block $b, b=2,\dots,B$, source terminal $i,i=1,2$, transmits the channel codeword $\xvec_i(\msg_{i,b},\svec_{i,b-1}, \qvec(h_i(\svec_{i,b-1})))$, where $\msg_{i,b} = f_i(\svec_{i,b}) \in \msgCal_i$.
At block $B+1$, source terminal $i,i=1,2$, transmits $\xvec_i(1,\svec_{i,B}, \qvec(h_i(\svec_{i,B})))$.

 Let $\tvec_1 = h_1(\avec_1) = h_2(\avec_2)$. At block $b=1$, the relay transmits $\xvec_3(\avec_1,\avec_2, \qvec(\tvec_1))$.
Assume that at block $b, b=2,\dots,B,B+1$, the relay has estimates $\tilde{\svec}_{i,b-1}$ of $\svec_{i,b-1}, i=1,2$, and let $\ttvec_{b-1} = h_1(\tsvec_{1,b-1}) = h_2(\tsvec_{2,b-1})$. The relay then transmits the
channel codeword $\xvec_3(\tilde{\svec}_{1,b-1},\tilde{\svec}_{2,b-1}, \qvec(\ttvec_{b-1}))$.

\subsection{Decoding} \label{annex:jointProofFlip_CP_Decoding}
The relay decodes the source sequences sequentially trying to reconstruct source block  $\svec_{i,b}, i=1,2$, at the end of channel block $b$ as follows: Let $(\tsvec_{1,b-1}, \tsvec_{2,b-1})$
be the estimates of $(\svec_{1,b-1}, \svec_{2,b-1})$ at the end of block $b-1$, and let $\ttvec_{b-1} \triangleq h_1(\tsvec_{1,b-1}) = h_2(\tsvec_{2,b-1})$.
The relay channel decoder at time $b$ decodes $(\msg_{1,b}, \msg_{2,b})$, 
by looking for a unique pair $(\tmsg_{1}, \tmsg_{2}) \in \msgCal_1 \times \msgCal_2$ such that:
\begin{align}
    \Big( & \tsvec_{1,b-1}, \tsvec_{2,b-1}, \ttvec_{b-1}, \qvec(\ttvec_{b-1}), \nonumber \\
		& \quad \xvec_1(\tmsg_{1}, \tsvec_{1,b-1},\qvec(\ttvec_{b-1})),\xvec_2(\tmsg_{2}, \tsvec_{2,b-1},\qvec(\ttvec_{b-1})), \nonumber \\
		& \qquad \xvec_3(\tsvec_{1,b-1}, \tsvec_{2,b-1}, \qvec(\ttvec_{b-1})), \wvec_{3,b-1}, \yvec_{3,b} \Big) \nonumber \\
		& \qquad \quad \in \styp (S_1,S_2,T,Q,X_1,X_2,X_3,W_3,Y_3).
    \label{eq:RelayJntFlipDecType_CP}
\end{align}

The decoded bin indices, denoted $(\tmsg_{1,b}, \tmsg_{2,b})$, are then given to the relay source decoder, which
estimates
$(\svec_{1,b}, \svec_{2,b})$ 
by looking for a unique pair of sequences $(\tsvec_{1},\tsvec_{2}) \in \mS_1^n\times\mS_2^n$
that satisfies $f_1(\tsvec_{1})= \tmsg_{1,b}$, $f_2(\tsvec_{2})= \tmsg_{2,b}$, and $(\tsvec_{1},\tsvec_{2},\wvec_{3,b}) \in \styp(S_1,S_2,W_3)$.
Denote the decoded sequences by $(\tsvec_{1,b},\tsvec_{2,b})$.

Decoding at the destination is done using backward decoding. The destination node waits until the end of channel block $B+1$.
It first tries to decode $\svec_{i,B}, i=1,2$, using the received signal at channel block $B+1$ and its side information $\wvec_{B}$.
Going backwards from the last channel block to the first, at channel block $b$ we assume that the destination has estimates $(\hsvec_{1,b+1},\hsvec_{2,b+1})$ of $(\svec_{1,b+1},\svec_{2,b+1})$ and consider decoding
of $(\svec_{1,b},\svec_{2,b})$. From $(\hsvec_{1,b+1}, \hsvec_{2,b+1})$ the destination finds the corresponding $(\hmsg_{1,b+1}, \hmsg_{2,b+1})$.
Then, the destination decodes $(\svec_{1,b},\svec_{2,b})$ by looking for a unique pair $(\hsvec_{1}, \hsvec_{2})\in\mS_1^n\times\mS_2^n$ such that:
\begin{align}
   & \big(\hsvec_{1}, \hsvec_{2}, \htvec, \qvec(\htvec), \xvec_1(\hmsg_{1,b+1}, \hsvec_{1}, \qvec(\htvec)),  \nonumber \\
    & \quad \xvec_2(\hmsg_{2,b+1}, \hsvec_{2}, \qvec(\htvec)), \xvec_3(\hsvec_{1}, \hsvec_{2}, \qvec(\htvec)), \wvec_{b}, \yvec_{b+1}\big) \nonumber \\
		& \qquad \in \styp (S_1,S_2,T,Q,X_1,X_2,X_3,W,Y).
\label{eq:DestFlipJntChanDecType_CP}
\end{align}

\noindent where $\htvec = h_1(\hsvec_{1}) = h_2(\hsvec_{2})$. Denote the decoded sequences by $\hsvec_{1,b}$ and $\hsvec_{2,b}$.

\ifthenelse{\boolean{SquizFlag}}{}{}

\ifthenelse{\boolean{SquizFlag}}{}{}

\ifthenelse{\boolean{IncludeJntOldProofs}}{}{}

\end{document}